\documentclass[a4paper,11pt,twoside,reqno]{amsart}

\usepackage[utf8]{inputenc}
\usepackage{xcolor}
\usepackage{hyperref}
\hypersetup{
    colorlinks,
    linkcolor={red!80!black},
    citecolor={green!80!black},
    urlcolor={blue!80!black},
    plainpages=false
}
\usepackage{amssymb,amsthm}
\usepackage[margin=1in]{geometry}
\usepackage{slashed}
\usepackage{doi}
\usepackage{caption}
\usepackage{adjustbox}

\newtheorem{theorem}{Theorem}[section]
\newtheorem{proposition}[theorem]{Proposition}

\theoremstyle{definition}
\newtheorem{definition}[theorem]{Definition}
\newtheorem{remark}[theorem]{Remark}

\newcommand{\tr}{\operatorname{Tr}}
\newcommand{\Ric}{\operatorname{Ric}}
\newcommand{\Id}{\operatorname{Id}}
\newcommand{\Scal}{\operatorname{Scal}}

\newcommand{\grad}{\operatorname{grad}}

\newcommand{\C}{{\mathbb{C}}}

\newcommand{\Z}{\mathbb{Z}}
\newcommand{\R}{\mathbb{R}}

\renewcommand{\epsilon}{\varepsilon}
\numberwithin{equation}{section}

\usepackage{booktabs,diagbox,hhline,multicol,multirow}
\usepackage{empheq}
\usepackage{tensor}
\usepackage[inline]{enumitem}
\usepackage{mathrsfs}
\setenumerate{font=\upshape,label=(\roman*),itemsep=0.1cm}

\newcommand{\diff}{\mathrm{d}}
\newcommand{\Y}{\mathrm{Y}}
\newcommand{\id}{\mathrm{Id}}
\newcommand{\iprod}{\mathbin{\lrcorner}}
\newcommand{\Fer}{\mathscr{F}} 
\newcommand{\Hig}{\mathscr{H}} 
\newcommand{\dirac}{\slashed{\mathrm{D}}{}}
\newcommand{\g}{\mathfrak{g}}
\renewcommand{\k}{\mathfrak{k}}
\newcommand{\h}{\mathfrak{h}}

\newcommand{\comment}[1]{}

\newcommand{\U}{\mathbf{U}}
\newcommand{\SU}{\mathbf{SU}}
\newcommand{\Spin}{\mathbf{Spin}}

\newcommand{\su}{\mathfrak{su}}
\newcommand{\spin}{\mathfrak{spin}}
\renewcommand{\u}{\mathfrak{u}}
\renewcommand{\d}{\mathfrak{d}}
\newcommand{\Ein}{\mathrm{Ein}}

\DeclareMathOperator{\Cl}{Cl}

\DeclareMathOperator{\SO}{\mathbf{SO}}

\DeclareMathOperator{\Ad}{Ad}

\makeatletter
\newsavebox\myboxA
\newsavebox\myboxB
\newlength\mylenA
\newcommand*\xoverline[2][0.75]{%
    \sbox{\myboxA}{$\m@th#2$}%
    \setbox\myboxB\null
    \ht\myboxB=\ht\myboxA%
    \dp\myboxB=\dp\myboxA%
    \wd\myboxB=#1\wd\myboxA
    \sbox\myboxB{$\m@th\overline{\copy\myboxB}$}
    \setlength\mylenA{\the\wd\myboxA}
    \addtolength\mylenA{-\the\wd\myboxB}%
    \ifdim\wd\myboxB<\wd\myboxA%
       \rlap{\hskip 0.5\mylenA\usebox\myboxB}{\usebox\myboxA}%
    \else
        \hskip -0.5\mylenA\rlap{\usebox\myboxA}{\hskip 0.5\mylenA\usebox\myboxB}%
    \fi}
\makeatother
\newcommand{\conj}[1]{\xoverline{#1}}

\usepackage[dvipsnames]{xcolor}

\title{Static spherically symmetric electroweak models with fermions and gravity}
\date{\today}
\author{Marko Sobak}
\address{University of Vienna, Faculty of Mathematics, Oskar-Morgenstern-Platz 1, 1090 Vienna, Austria}
\email{marko.sobak@univie.ac.at}
\thanks{The author gratefully acknowledges the support of the Austrian Science Fund (FWF) through the project \emph{The standard model as a geometric variational problem} (DOI: 10.55776/P36862). I am grateful to Jan Heck for useful comments on the manuscript. Many thanks also to Volker Branding and Adam Lindström for the continuous support and listening to my rants about notations and conventions.}

\begin{document}

\begin{abstract}
    In this article we provide a systematic reduction to static spherical symmetry for (classical) gauge theoretic models with structure group $\SU(2) \times \U(1)$, coupled to gravity.
    In particular, we derive the most general static spherically symmetric ansatz for the Einstein-Yang-Mills-Higgs-Dirac-Yukawa equation system.
    This model is reminiscent of the lepton sector of the Standard Model coupled to gravity. 
    The resulting total system of ordinary differential equations is rather large, so that finding global solutions is non-trivial.
    For this reason, an open-source interactive solver GUI written in C++ is provided together with the article.
\end{abstract}

\maketitle
\tableofcontents

\section{Introduction}

Gauge theory is the backbone of most modern physics, most notably the Standard Model of particle physics (SM), as well as general relativity (GR), whose associated fields are at a classical level most naturally viewed as objects living on a principal fiber bundle over the spacetime manifold.
From a physics perspective, these theories are formulated as variational problems, in the sense that they are modelled via a Lagrangian, i.e.\ an action functional. 
This makes it natural to also study the associated classical Euler-Lagrange equations, which form a system of nonlinear coupled PDEs, cf.\ \S \ref{sec-e-l-equations-general}. 

Due to the complexity of these PDEs, it is beneficial to further simplify them by imposing additional symmetries on the fields, which ultimately reduces the number of variables in the equations.
Possibly the most prominent choice of symmetry is that of spherical symmetry, which has been studied in literature in many different settings and contexts, particularly for isolated sectors of the model.
Let us give a brief and by no means exhaustive review of the relevant literature, particularly for the static setting.
The static spherically symmetric Einstein-Yang-Mills equations with gauge group $\SU(2)$ were extensively studied in the 90s, see the works of
Bartnik, McKinnon \cite{bart-mckin},
Bizoń \cite{bizon-colored-bh},
Smoller, Wasserman, Yau, McLeod \cite{smol-wass-1},
Breitenlohner, Forgacs, Maison \cite{breit-forg-mais}.
For a more extensive review, one can refer to the paper \cite{volkov-galtsov} of Volkov and Gal'tsov.
More recently, Gervalle and Volkov \cite{volkov-electroweak, volkov-electroweak-blackholes} performed an extensive numerical analysis of the Einstein-Yang-Mills-Higgs system with gauge group $\SU(2) \times \U(1)$.
Concerning the coupling of the Dirac sector to gravity and the Maxwell/Yang-Mills sector, there are the works of Finster, Smoller, Yau \cite{finster-smoller-yau-edym-solutions, finster-smoller-yau-edm-nonexistence, finster-smoller-yau-edym-bh-nonexistence}.
There have also been several works of Kain on the coupling between Einstein-Maxwell/Yang-Mills and Dirac sectors, see e.g.\ \cite{kain-edm-wormholes}.

When studying models with symmetries, it is common in the literature to employ different formalisms to derive a consistent ansatz for the fields, see e.g.\ the work of Balasin, Böhmer, Grumiller \cite{sph-sym-standard-model} for an exposition, where also various notions of spherical symmetry are covered.
What appears to be less known is that a general framework for studying gauge theoretic models with symmetries was developed by several authors at the turn of the 21st century, including Brodbeck \cite{brodbeck}, Künzle, Oliynyk \cite{kunzle}, Harnad, Shnider, Vinet \cite{harnad}, see also \S \ref{sec-symmetry-general} of this paper for a review.
This framework is incredibly convenient for two reasons.
Firstly, it provides unified definitions, in terms of group actions on principal bundles, of what it means for the geometric objects to be symmetric with respect to arbitrary group actions (not necessarily corresponding to spherical symmetry).
Secondly, it provides classifications of such symmetric objects in terms of simple algebraic conditions.
This allows one to systematically determine the most general symmetric reduction for essentially any gauge theoretic model, in a rather clean algebraic manner.

The main goal of this paper is to use this framework to obtain the most general spherically symmetric (in the above sense) gauge theoretic model including the following features:
\begin{itemize}
    \item a spherically symmetric base spacetime $(M,g)$,
    \item a non-abelian Yang-Mills sector modelling electroweak bosons,
    \item a Higgs sector with potential, modelling the mass mechanism for electroweak bosons,
    \item a Dirac sector, modelling fermions as twisted chiral spinor fields,
    \item a Yukawa coupling, modelling the mass mechanism of fermions via the Higgs field.
\end{itemize}
To the author's knowledge, no fully consistent spherically symmetric gauge theoretic model in the mathematical literature so far has checked all of these items in conjuction, in particular the last two.
In fact, the Dirac sector for most of the known spherically symmetric models in literature (in particular the above mentioned ones) either drop the chiral structure of the model and consequently also the Yukawa coupling, or they work with ansätze defined in terms of spin-weighted harmonics which would not be considered to be spherically symmetric in the context of the general framework used here (though they would be considered to be weakly spherically symmetric in the sense of \cite{sph-sym-standard-model}).

The existence of a consistent ansatz supporting all these requirements depends heavily on the representations for the associated vector bundles used in the model. 
For specific choices of the representation parameters, one can obtain all of the above listed models from literature as special cases, see \S \ref{sec-relation-to-literature} for a detailed comparison.
More importantly, on can also choose the representation parameters in such a way that the model resembles the lepton sector of the SM, albeit unfortunately does not match it exactly, as it features different $\U(1)$ charges for the fields. 

Studying a symmetry reduced model like the one presented in this work is instructive since the interactions and the finer structure of the model generally becomes simpler to see in the symmetry reduced form, in particular since all of the involved bundles, actions and fields can be written down rather explicitly.
In view of this, the model might provide a good tool for educational purposes in gauge theory, but also a playground for a more detailed study of axiomatic quantization.
That being said, we would like to emphasize again that this paper treats the model only as a classical field theory, i.e.\ we do not consider quantization in any form.

Aside from just deriving an ansatz for the model and its underlying structures, the more important purpose of this paper is to derive the associated static Euler-Lagrange equations, which reduce to ODEs as static spherical symmetry leaves only one degree of freedom.
The resulting equation system is rather large (4 complex-valued first-order equations, 2 complex-valued second-order equations, 4 real-valued second-order equations, and 3 real-valued constraints), cf.\ \S \ref{sec-total-eqs}.
Thus one may be inclined to search for further reductions for the ansatz, although the usual simplification tricks seem to fail if one does not want to lose generality and in particular drop some of the items from the above discussed checklist.
Vaguely speaking, the reason behind this is the fact that the chiral subbundles of the Dirac sector are twisted by non-isomorphic vector bundles (like in the SM), which leads to a certain degree of structural asymmetry, see \S \ref{sec-further-reductions} for a more detailed discussion.
This makes the problem of finding global solutions to the problem highly non-trivial.
On the other hand, since this model features some more intricate couplings between the fields compared to the simplified models mentioned above, it is not unreasonable to expect that this could lead also to more interesting dynamics of the solutions. 

The paper focuses on providing a detailed description of the structure of the model and the derivation of the ansatz and associated equations in detail, rather than the analysis of the equations themselves.
Nevertheless, some essential properties of the system, such as the propagation of constraints and constants of motion, are discussed for the system in full generality (i.e.\ for arbitrary valid representation parameters), as those are necessary for setting up the equations as a well-posed initial value problem.

In addition to the manuscript, an interactive solver GUI for the system of Euler-Lagrange equations is provided.
The program is open-source and can be obtained from the git repository
\begin{center}
    \href{https://codeberg.org/msobak/SphSymSM.git}{\texttt{https://codeberg.org/msobak/SphSymSM.git}}
\end{center}
or by messaging the author.
The code is written in C++ and uses an adaptive Runge-Kutta-Fehlberg 4(5) scheme to produce the solution in real time, based on the input parameters from the GUI.
More information can be found in the documentation of the linked repository.

\textbf{AI declaration.}
The mathematical contents of this work, in particular all the derivations and calculations, were done by hand by the author. 
AI was not used in the writing process.

The attached solver code is also human-written by the author.
The Claude chatbot (Sonnet 4--5) was used while writing the solver code, though mostly just for diagnosing bugs and referencing the C++ documentation, and not for automatic code generation.

\section{Gauge theories with gravity}

In this section we briefly review the geometric structure of geometric gauge theories (in particular inspired by the Standard Model of particle physics) coupled to gravity.
The goal of this section is ultimately to present the equations that will be studied in the rest of the work.
For a more comprehensive review of gauge theory, one can refer to \cite{hamilton}, while for the Einstein equations we recommend \cite{misner-thorne-wheeler}.

\subsection{Lagrangian}
\label{sec-sm-lagrangian}

Throughout the section, we assume $(M,g)$ to be an oriented and time-oriented four-dimensional spin Lorentzian manifold with the \emph{mostly-plus} signature convention. We use the same conventions as in \cite{bs25} for forms, inner products, curvature, signs of operators, etc., so we do not restate them here. 

Let $G$ be a compact Lie group equipped with an Ad-invariant positive-definite inner product on its Lie algebra $\g$.
Let $G \to P \xrightarrow\pi M$ be a principal bundle equipped with a connection $\omega \in \Omega^1(P,\g)$, and let $F_\omega \in \Omega^2(M, \Ad P)$ be its curvature form.
The connection $\omega$ encodes information about \emph{bosons}.
Let also $W$ be a complex linear space equipped with a Hermitian inner product and let $\rho : G \to \mathbf{U}(W)$ be a complex representation. We consider the associated bundle $\Hig = P \times_\rho W$ whose sections are called \textit{Higgs fields}.
The connection $\omega$ induces a covariant derivative on $\Hig$ that is denoted by $\nabla_{\omega}$.
Since the representation $\rho$ is unitary, the inner product on $W$ also induces a metric on the Higgs field bundle $\Hig$.
The \textit{Yang-Mills-Higgs (bosonic) Lagrangian density} is given by 
\begin{align*}
    (\omega, \Phi) \mapsto &-\left(|F_\omega|^2 + |\nabla_{\omega} \Phi|^2 + U(\Phi)\right)
    \\
    =& -\frac12 \langle (F_\omega)_{\mu\nu}, (F_\omega)^{\mu\nu}\rangle - \langle (\nabla_\omega\Phi)_\mu, (\nabla_\omega\Phi)^\mu \rangle - U(\Phi),
\end{align*}
where $U: \Hig \to \R$ is the smooth $G$-invariant potential
\begin{equation}\label{eq-mexican-hat-potential}
    U(\Phi) = \frac12 (\lambda^2 - |\Phi|^2)^2,
\end{equation}
for a constant $\lambda > 0$.

Now let $V$ be another complex linear space equipped with a Hermitian inner product.
Let $\chi : G \to \mathbf{U}(V)$ be a complex representation and define the associated bundle $\mathscr{S} = P \times_\chi V$. We consider the twisted spinor bundle 
\begin{equation*}
\Fer = \Sigma M \otimes \mathscr{S}.
\end{equation*}
The connection $\omega$ induces a covariant derivative on $\mathscr{S}$, which together with the spinor covariant derivative induces a \textit{twisted spinor covariant derivative} which will be denoted by $\nabla_\omega$.
The \textit{twisted Dirac operator} is the map
\begin{equation*}
    \dirac_{\omega} : \Gamma(\Fer) \to \Gamma(\Fer), \qquad 
    \dirac_{\omega}\Psi = \eta^{\mu\nu} e_\mu \cdot (\nabla_{\omega} \Psi)(e_\nu). 
\end{equation*}
Note that \(i\dirac_\omega\) is self-adjoint with respect to the \(L^2\)-inner product, since Clifford multiplication is symmetric (cf.\ \S \ref{sec-important-matrices} below for the conventions).

The four-dimensional spinor bundle splits into chiral subbundles $\Sigma M = \Sigma_+ M \oplus \Sigma_- M$.
We assume that the representation $\chi$ also splits so that $V = V_+ \oplus V_-$ is an orthogonal decomposition and $\chi = \chi_+ \oplus \chi_-$ where $\chi_\pm : G \to \mathbf{U}(V_\pm)$ are representations (in general non-isomorphic).
We also define the associated bundles $\mathscr{S}_\pm = P \times_{\chi_\pm} V_\pm$, as well as the \textit{twisted chiral/mixed spinor bundles}, respectively, as 
\begin{equation}
\label{eq-twisted-chiral-bundle}
    \Fer_+ = (\Sigma_+M \otimes \mathscr{S}_+) \oplus (\Sigma_-M \otimes \mathscr{S}_-), \qquad \Fer_- = (\Sigma_+M \otimes \mathscr{S}_-) \oplus (\Sigma_-M \otimes \mathscr{S}_+).
\end{equation}
Sections $\Psi \in \Gamma(\Fer_+)$ are called \textit{twisted chiral spinors}.
The twisted Dirac operator restricts to a mapping
\begin{equation*}
    \dirac_{\omega} : \Gamma(\Fer_\pm) \to \Gamma(\Fer_\mp).
\end{equation*}

A \textit{Yukawa map} is an $\R$-linear map $\Y : W \to \mathfrak{u}(V), \, w \mapsto \Y_w$, such that
\begin{enumerate}
    \item $\Y_w (V_\pm) \subset V_\mp$ for each $w\in W$,
    \item $\Y$ is $G$-equivariant in the sense that 
    \begin{equation}\label{eq-yukawa-equivariant}
        \Y_{\rho(g)w}(\chi(g)v) = \chi(g) \Y_w v
    \end{equation}
    for each $g\in G$, $w\in W$ and $v \in V$. 
\end{enumerate}
By the equivariance property (\ref{eq-yukawa-equivariant}), a Yukawa map induces an $\R$-tensorial map taking Higgs fields to couplings between $\Fer_\pm$, i.e.\
\begin{equation*}
    \Gamma(\Hig) \ni \Phi \mapsto \Y_\Phi : \Gamma(\Fer_\pm) \to \Gamma(\Fer_\mp),
\end{equation*}
such that $i\Y_\Phi$ is self-adjoint (see \cite{bs25} for more details).
Thus, the total Lagrangian density for fermions is given by
\begin{equation*}
    \Psi \mapsto  -\Re \langle \Psi, i(\dirac_{\omega} + \Y_\Phi) \Psi \rangle.
\end{equation*}

The total gauge-theoretic matter Lagrangian $L^G$ is thus given by
\begin{equation*}
    L^{G}: (\omega,\Phi,\Psi) \mapsto 
    -\left(\vert F_\omega \vert^2  + \vert \nabla_{\omega} \Phi \vert^2 + U(\Phi) + 
    \Re\langle \Psi, i(\dirac_{\omega} + \Y_\Phi) \Psi \rangle\right),
\end{equation*}
where $\omega \in \Omega^1(P,\g)$ is a connection, $\Phi \in \Gamma(\Hig)$ is a Higgs field, and $\Psi \in \Gamma(\Fer_+)$ is a twisted chiral spinor field.

Finally, we couple the matter Lagrangian $L^G$ with the Einstein-Hilbert Lagrangian $L^E$, to get the full Lagrangian density
\begin{align*}
    L : (g, \omega, \Phi, \Psi)
    \mapsto& \, 
    L^{E}(g) + L^{G}(g,\omega,\Phi,\Psi)
    \\
    =& \, \Scal_g - \left(\vert F_\omega \vert^2  + \vert \nabla_{\omega} \Phi \vert^2 + \frac12 (\lambda^2 - |\Phi|^2)^2 + 
    \Re\langle \Psi, i(\dirac_{\omega} + \Y_\Phi) \Psi \rangle\right),
\end{align*}
where $\Scal_g$ is the scalar curvature of $(M,g)$.
This Lagrangian models a classical coupling between gravity (i.e.\ the metric $g$) and matter (i.e.\ the fields $(\omega, \Phi, \Psi)$).

\subsection{Euler-Lagrange equations}
\label{sec-e-l-equations-general}

The matter field equations are obtained by varying the integral of the Lagrangian $L$ with respect to $(\omega, \Phi, \Psi)$, while a variation with respect to $g$ leads to the Einstein field equations.
The variations are well-known in the literature and have been presented in many places.
For a derivation using the same notations and conventions as in the present paper, we refer to \cite{bs25, bls26}.
The resulting equations are given by
\begin{subequations}
\begin{empheq}[left=\empheqlbrace]{align}
    \diff_\omega^\ast  F_\omega + \Re \langle \nabla_{\omega}\Phi \otimes \rho_*\Phi \rangle - \tfrac12\Im \langle \id\cdot\Psi \otimes \chi_*\Psi \rangle &= 0 , \label{eq-yang-mills}
    \\[0.2cm]
    \Box_{\omega} \Phi + (\lambda^2 - |\Phi|^2)\Phi  - \langle \Psi, i\Y^-  \Psi \rangle &= 0 , 
    \label{eq-higgs}
    \\[0.2cm]
    \dirac_{\omega} \Psi + \Y_\Phi\Psi &= 0, 
    \label{eq-dirac}
    \\[0.2cm] 
    -\Ein_g + \mathfrak{T}_g[\omega, \Phi, \Psi] &= 0.
    \label{eq-einstein}
\end{empheq}
\end{subequations}
Here, $\Box_\omega = -\nabla^*_\omega\nabla_\omega = \tr_g(\nabla_\omega^2)$ is the connection wave operator, and the currents are explicitly given by
\begin{align}
    \Re \langle \nabla_{\omega}\Phi \otimes \rho_*\Phi \rangle &= \Re \langle (\nabla_{\omega}\Phi)_\mu, \, \rho_*(\xi_a)\Phi \rangle \, e^\mu \otimes \xi_a,
    \label{eq-higgs-current}
    \\[0.1cm]
    \Im \langle \id\cdot\Psi \otimes \chi_*\Psi \rangle &= \frac12 \Im \langle e_\mu \cdot \Psi, \, \chi_*(\xi_a)\Psi \rangle\, e^\mu \otimes \xi_a,
    \label{eq-dirac-current}
    \\[0.1cm]
    \langle \Psi, i\Y^-  \Psi \rangle &= \frac{1}{2} \langle \Psi, (i\Y_{W_k} - \Y_{iW_k}) \Psi \rangle \, W_k,
    \label{eq-yukawa-current}
\end{align}
where $e_\mu$ is a semi-orthonormal frame for $\SO^+(TM)$, $\xi_a$ is an orthonormal frame for $\Ad P$, and $W_k$ is an orthonormal frame for $\Hig$.
Furthermore, $\Ein_g = \Ric_g - \frac12 \Scal_g g$ is the Einstein tensor and $\mathfrak{T}$ denotes the energy-momentum tensor, given by
\begin{equation*}
    \mathfrak{T} = \mathfrak{T}^{\mathrm{YM}} + \mathfrak{T}^{\mathrm{H}} + \mathfrak{T}^{\mathrm{D}},
\end{equation*}
where
\begin{align*}
\mathfrak{T}^{\mathrm{YM}}(X,Y)
&= \langle X \iprod F_\omega, Y \iprod F_\omega \rangle - \frac12 |F_\omega|^2 g(X,Y),
\\[0.1cm]
\mathfrak{T}^{\mathrm{H}}(X,Y)
&= \Re \langle (\nabla_\omega \Phi)(X), (\nabla_\omega \Phi)(Y) \rangle - \frac12 \left(|\nabla_\omega\Phi|^2 + U(\Phi)\right) g(X,Y),
\\[0.1cm]
\mathfrak{T}^{\mathrm{D}}(X,Y) 
&= \frac14 \, \Re \langle \Psi, iX \cdot (\nabla_\omega\Psi)(Y) + iY \cdot (\nabla_\omega\Psi)(X) \rangle,
\end{align*}
for vector fields $X,Y \in \Gamma(TM)$.
The Einstein-SM equations (\ref{eq-yang-mills}--\ref{eq-dirac}) are the equations we will study throughout the rest of the work.

\section{Representations}

To study the equations (\ref{eq-yang-mills}--\ref{eq-dirac}) and the model in closer detail, one of course needs to fix all the representations for the particle spaces.

\subsection{Some important matrices}
\label{sec-important-matrices}

Let $\sigma_k \in \C^{2\times 2}$ be the Pauli matrices, given by
\begin{equation}\label{eq-pauli-matrices}
    \sigma_1 = \begin{bmatrix}
        0 & 1\\
        1 & 0
    \end{bmatrix},
    \quad
    \sigma_2 = \begin{bmatrix}
        0 & -i\\
        i & 0
    \end{bmatrix},
    \quad
    \sigma_3 = \begin{bmatrix}
        1 & 0\\
        0 & -1
    \end{bmatrix}.
\end{equation}
The Pauli matrices induce a natural basis for the Lie algebra $\su(2) \cong \spin(3)$ given by
\begin{equation*}
    \tau_k = -\frac{i}{2} \sigma_k, \qquad [\tau_j,\tau_k] = \epsilon_{jk\ell} \tau_\ell.
\end{equation*}

For the Clifford relations with respect to a given inner product, we use the convention
\begin{equation*}
    v \cdot w + w \cdot v = -2\langle v,w\rangle.
\end{equation*}
We let $\Cl(r,s)$ be the Clifford algebra of the semi-Euclidean space $\R^{r,s}$ with metric $\eta$ of signature $(r,s)$, where $r$ is the number of minus signs and $s$ the number of plus signs. 
We consider the following two important cases:
For the spatial Clifford algebra $\Cl(3) = \Cl(0,3)$, we employ the representation
\begin{equation*}
    \beta_k = -i\sigma_k = 2\tau_k
\end{equation*}
on $\C^2$. It shall be useful to note that 
\begin{equation*}
    \exp(s\beta_k) = \cos(s) I_2 + \sin(s)\beta_k.
\end{equation*}
For the spacetime Clifford algebra $\Cl(1,3)$ (recall that we are using the mostly-plus convention for the spacetime metric throughout the manuscript), we employ the Weyl representation on $\C^4$ given by the matrices  
\begin{equation}\label{eq-gamma-matrices}
    \gamma_0 = \begin{bmatrix}
        0 & I\\
        I & 0
    \end{bmatrix},
    \qquad
    \gamma_k = \begin{bmatrix}
        0 & \sigma_k\\
        -\sigma_k & 0
    \end{bmatrix}.
\end{equation}
The volume element (or the fifth gamma matrix) is 
\begin{equation*}
    \gamma_5 = -i\gamma_0\gamma_1\gamma_2\gamma_3 = i\gamma^0\gamma^1\gamma^2\gamma^3 = \begin{bmatrix}
        I & 0\\
        0 & -I
    \end{bmatrix}.
\end{equation*}
The Weyl representation has the beneficial feature of emphasizing the chiral structure of spinors, and in particular we can split the components of $\psi$ as $\psi = (\psi_+, \psi_-)$, where $\psi_\pm \in \C^2$ are spinors of positive/negative chirality, i.e.\ they span the $\pm 1$-eigenspaces of $\gamma_5$. 
We will denote the standard basis for $\C^4$ by $\u_\pm, \d_\pm$ to emphasize the chirality, so that
\begin{equation}\label{eq-spinor-up-down-chiral-basis}
    \u_+ = e_1, \quad \d_+ = e_2, \quad \u_- = e_3, \quad \d_- = e_4,
\end{equation}
where $e_i$ is the standard basis for $\C^4$.
We summarize the Clifford action on this basis in Table \ref{tab:clifford}.
The inner product on $\mathbb{C}^4$ defined by the formula
\begin{equation}\label{eq-spinor-inner-product}
    \langle \psi, \phi \rangle = \psi^\dagger \gamma_0 \phi,
\end{equation}
is invariant under the induced action of $\Spin^+(1,3)$, and can thus be used to define an inner product on spinors. 
We note that this inner product has the property that Clifford multiplication is symmetric, i.e.\ 
$\langle \gamma_\mu \psi, \phi \rangle = \langle \psi, \gamma_\mu \phi \rangle$.

We recall the group embedding
\begin{equation}\label{eq-SU2-Spin13-embedding}
    \SU(2)\cong\Spin(3) \to \Spin^+(1,3), \qquad v_1\cdots v_{2n} \mapsto (ie_0 \cdot v_1) \cdots (ie_0 \cdot v_{2n}) = v_1 \cdots v_{2n},
\end{equation}
for vectors with unit norm $|v_k|=1$, where we use the complexified algebras.
In particular, with respect to the Clifford representations above, this embedding maps
\begin{align}
    \SU(2) \supset \U(1) \ni&\, \exp(s\tau_3) = \cos\left(\frac{s}{2}\right) I_2 + \sin\left(\frac{s}{2}\right) \beta_3
    \nonumber
    \\
    &= \left(\cos\left(\frac{s}{4}\right) \beta_1 + \sin\left(\frac{s}{4}\right) \beta_2\right)\left(-\cos\left(\frac{s}{4}\right) \beta_1 + \sin\left(\frac{s}{4}\right) \beta_2\right)
    \nonumber
    \\
    &\mapsto 
    \left(\cos\left(\frac{s}{4}\right) \gamma_1 + \sin\left(\frac{s}{4}\right) \gamma_2\right)\left(-\cos\left(\frac{s}{4}\right) \gamma_1 + \sin\left(\frac{s}{4}\right) \gamma_2\right)
    \nonumber
    \\
    &= \cos\left(\frac{s}{2}\right) I_4 + \sin\left(\frac{s}{2}\right) I_2 \otimes \beta_3
    = I_2 \otimes \exp(s\tau_3) \in \Spin^+(1,3),
    \label{eq-U(1)-inside-spin}
\end{align}
where we use $\otimes$ for the Kronecker product of matrices.

We define the maps $M_{\mu\nu} : \R^{1,3} \to \R^{1,3}$ via
\begin{equation}\label{eq-so13-basis}
    M_{\mu\nu}: e_\lambda \mapsto \eta_{\mu\lambda} e_\nu - \eta_{\nu\lambda} e_\mu.
\end{equation}
The set $\{M_{\mu\nu}, \; \mu < \nu\}$ then forms a basis for the Lie algebra $\mathfrak{so}(1,3)$.
The natural double covering $\rho : \Spin^+(1,3) \to \SO^+(1,3)$ induces the isomorphism $\rho_* : \mathfrak{spin}(1,3) \to \mathfrak{so}(1,3)$, which satisfies
\begin{equation}\label{eq-double-cover-lie-algebra}
    \rho_* : \tfrac12 \gamma_\mu\gamma_\nu \mapsto M_{\mu\nu},
\end{equation}
i.e.\ it maps the natural basis $\{\tfrac12 \gamma_\mu\gamma_\nu, \; \mu < \nu\}$ for $\mathfrak{spin}(1,3)$ to the natural basis for $\mathfrak{so}(1,3)$.
The action of this basis on the spinor space is displayed in Table \ref{tab:spin-action}.

\begin{table}[t]
\begin{adjustbox}{center}
\begin{minipage}[c]{0.5\textwidth}
\centering
\footnotesize
\renewcommand{\arraystretch}{1.5}
    \begin{tabular}{|c||c|c|c|c|}
        \hline
        $\gamma_\mu\psi$ & $\psi=\u_+$ & $\psi=\d_+$ & $\psi=\u_-$ & $\psi=\d_-$ \\\hline\hline
        $\mu=0$ & $\u_-$ & $\d_-$ & $\u_+$ & $\d_+$ \\\hline
        $\mu=1$ & $-\d_-$ & $-\u_-$ & $\d_+$ & $\u_+$ \\\hline
        $\mu=2$ & $-i\d_-$ & $i\u_-$ & $i\d_+$ & $-i\u_+$ \\\hline
        $\mu=3$ & $-\u_-$ & $\d_-$ & $\u_+$ & $-\d_+$ \\\hline
    \end{tabular}
    \captionof{table}{Clifford multiplication in the Weyl representation}
    \label{tab:clifford}
\end{minipage}
\begin{minipage}[c]{0.5\textwidth}
\centering
\footnotesize
\renewcommand{\arraystretch}{1.5}
\begin{tabular}{|c||c|c|c|c|}
        \hline
        $\frac12 \gamma_\mu\gamma_\nu\psi$  & $\psi=\u_+$         & $\psi=\d_+$       & $\psi=\u_-$           & $\psi=\d_-$ \\\hline\hline
        $(\mu,\nu)=(0,1)$                   & $-\frac12 \d_+$     & $-\frac12 \u_+$   & $\frac12 \d_-$        & $\frac12 \u_-$  \\\hline
        $(\mu,\nu)=(0,2)$                   & $-\frac{i}{2} \d_+$ & $\frac{i}{2}\u_+$ & $\frac{i}{2} \d_-$    & $-\frac{i}{2}\u_-$  \\\hline
        $(\mu,\nu)=(0,3)$                   & $-\frac12 \u_+$     & $\frac12 \d_+$    & $\frac12 \u_-$        & $-\frac12\d_-$ \\\hline
        $(\mu,\nu)=(1,2)$                   & $-\frac{i}{2}\u_+$  & $\frac{i}{2}\d_+$ & $-\frac{i}{2}\u_-$    & $\frac{i}{2}\d_-$
        \\\hline
        $(\mu,\nu)=(1,3)$                   & $-\frac12 \d_+$     & $\frac12 \u_+$    & $-\frac12 \d_-$       & $\frac12 \u_-$
        \\\hline
        $(\mu,\nu)=(2,3)$                   & $-\frac{i}{2}\d_+$  & $-\frac{i}{2}\u_+$& $-\frac{i}{2}\d_-$    & $-\frac{i}{2}\u_-$
        \\\hline
    \end{tabular}
    \captionof{table}{Action of the spin algebra in the Weyl representation}
    \label{tab:spin-action}
\end{minipage}
\end{adjustbox}
\end{table}

\subsection{Particle representations}
\label{sec-particle-representations}

Throughout the paper, we work with the structure group $G = \SU(2) \times \U(1)$.
We define the matrices
\begin{equation*}
    \tau_i = -\frac{i}{2}\sigma_i,
\end{equation*}
which form a basis for the Lie algebra $\mathfrak{su}(2)$.
We also need another linearly independent element to span the Lie algebra $\mathfrak{u}(1)$ of the second factor, and for this we set
\begin{equation*}
    \tau_4 = -\frac{i}{2} \in \mathfrak{u}(1) \cong i\R.
\end{equation*}
We then define an Ad-invariant inner product on $\mathfrak{su}(2) \oplus \mathfrak{u}(1)$ by demanding that $\tau_1,\ldots,\tau_4$ are orthonormal, explicitly we take
\begin{equation}\label{eq-su2-u1-inner-product}
\langle Z, W \rangle_{\mathfrak{su}(2)} = 2\tr(Z^\dagger W) = -2\tr(ZW),
\qquad
\langle z,w \rangle_{\mathfrak{u}(1)} = 4\conj{z}w = -4zw,
\end{equation}
which induce a natural inner product on $\su(2) \oplus \u(1)$ as well.

\begin{remark}
Here and throughout the rest of the manuscript we abuse the notation slightly, since the basis elements $\tau_\ell$ for $\ell =1,2,3$ should be understood as $\tau_\ell \oplus 0$, while $\tau_4$ should be understood as $0 \oplus \tau_4$, as elements of $\mathfrak{su}(2) \oplus \mathfrak{u}(1)$.
Furthermore, both of the inner products in \eqref{eq-su2-u1-inner-product} can in principle be rescaled by arbitrary positive constants, which should then be viewed as additional parameters for the model, cf.\ the discussion in \cite[\S 7.4]{hamilton}.
In this case the chosen basis elements $\tau_k$ would not be orthonormal however, so we do not pursue this generality for simplicity, since we will have enough parameters to deal with as-is.
\end{remark}

Irreducible representations of a direct product group are known to be tensor products of the irreducible representations of its factors.
Firstly, the irreducible representations of $\U(1)$ are given by integral powers, i.e.\ for each $q \in \Z$ we have the representation $\alpha_q: e^{it} \mapsto e^{iqt}$ on $\C$, or infinitesimally, $\tau_4 \mapsto q\tau_4$.
For $\SU(2)$, the irreducible representations are classified by integers $p \geq 0$, with the $n$-th representation being  $\beta_p: \SU(2) \to \U(X^{p+1})$ of dimension $p+1$ with an orthonormal basis $v_0, \dots, v_p$ of $X^{p+1}$ satisfying
\begin{align*}
    (\beta_p)_*(\tau_1)v_k &= -\frac{i}{2}\sqrt{k(p-k+1)}v_{k-1} - \frac{i}{2}\sqrt{(k+1)(p-k)}v_{k+1}\\
    (\beta_p)_*(\tau_2)v_k &= -\frac{1}{2}\sqrt{k(p-k+1)}v_{k-1} + \frac{1}{2}\sqrt{(k+1)(p-k)}v_{k+1}\\
    (\beta_p)_*(\tau_3)v_k &= -\frac{i}{2}(p-2k)v_k.
\end{align*}
We will refer to the representation $\rho_{(p,q)} = \xi_p \otimes \alpha_q$ of $\SU(2) \times \U(1)$ as the \emph{$(p,q)$-representation} of $G$, and $(p,q)$ as the \emph{representation parameters}.
Then explicitly
\begin{equation}
\label{eq-representation-actions}
\begin{split}
    (\rho_{(p,q)})_*(\tau_1)v_k &= -\frac{i}{2}\sqrt{k(p-k+1)}v_{k-1} - \frac{i}{2}\sqrt{(k+1)(p-k)}v_{k+1}\\
    (\rho_{(p,q)})_*(\tau_2)v_k &= -\frac{1}{2}\sqrt{k(p-k+1)}v_{k-1} + \frac{1}{2}\sqrt{(k+1)(p-k)}v_{k+1}\\
    (\rho_{(p,q)})_*(\tau_3)v_k &= -\frac{i}{2}(p-2k)v_k\\
    (\rho_{(p,q)})_*(\tau_4)v_k &= -\frac{i}{2}qv_k.
\end{split}
\end{equation}


\begin{definition}
\label{def-parameters}
Let $p,q,d_\pm,h_\pm$ be integers, with $p,d_\pm \geq 0$.  
\begin{itemize}[itemsep=0.1cm]
\item The \emph{Higgs vector space} $W = W_{(p,q)}$ is the $(p+1)$-dimensional space associated to the $(p,q)$-representation, with orthonormal basis $w_k \in W$,
\item The \emph{$\pm$-fermion vector space} $V_\pm = V_{(d_\pm, h_\pm)}$ is the $(d_\pm+1)$-dimensional space associated to the $(d_\pm, h_\pm)$-representation, with orthonormal basis $v^{\pm}_k \in V_\pm$. We also denote the total space by $V = V_+ \oplus V_-$.
\end{itemize}
\end{definition}

\subsection{Yukawa coupling}

Let us now derive an ansatz for a Yukawa map to be used in the Yukawa coupling.
We recall from \S \ref{sec-sm-lagrangian} that this is an $\R$-linear map $\Y:W\to\u(V)$ such that $\Y_w(V_\pm) \subset V_{\pm}$ and satisfies the equivariance law \eqref{eq-yukawa-equivariant}.
First, we note that we can decompose any Yukawa map as
\begin{equation*}
    \Y_w v = \Y_w (v^+\oplus v^-) = \Y_w^- v^- \oplus \Y_w^+ v^+,
    \qquad
    \Y_w^- = -(\Y_w^+)^\dagger
\end{equation*}
where $\Y_w^\pm : V_\pm \to V_\mp$ and the last condition ensures that $\Y_w \in \u(V)$.
In particular it suffices to determine $\Y_w^+ : V_+ \to V_-$.
For simplicity, we can assume that $w\mapsto \Y_w^+$ is $\C$-linear in $w$ (note that the full map $\Y_w$ is then still just $\R$-linear in $w$).
Then $\Y^+$ can be identified with an element 
\begin{equation*}
    \Y^+ \in W^* \otimes V_+^* \otimes V_-,
\end{equation*}
and the equivariance condition \eqref{eq-yukawa-equivariant} is equivalent to requiring that $\Y^+ \in \ker(\xi)$, where $\xi$ is the induced representation of $\g$.
We write
\begin{equation*}
    \Y^+ = \Y_{k,\ell}^{j} \, w^k \otimes v_+^{\ell} \otimes v^-_{j},
\end{equation*}
where $w_k,v^+_a,v^-_b$ are bases of $W, V_+, V_-$ respectively, and $w^k, v_+^a, v_-^b$ are their dual bases.
Then by \eqref{eq-representation-actions}, the condition $\Y^+ \in \ker(\xi)$ is solved as follows:
\begin{enumerate}[itemsep=0.1cm]
\item $\xi(\tau_4)\Y^+ = 0$ if and only if $q + h_+ - h_- = 0$,
\item $\xi(\tau_3)\Y^+ = 0$ if and only if $\Y_{k,\ell}^j =0$ whenever $p+d_+-d_- \not= 2(k+\ell-j)$,
\item $\xi(\tau_1)\Y^+ = \xi(\tau_2)\Y^+ = 0$ if and only if the components satisfy the recurrence relations
\begin{align*}
    0&= \sqrt{(k+1)(p-k)}\Y_{k+1,\ell}^j + \sqrt{(\ell+1)(d_+-\ell)}\Y_{k,\ell+1}^j - \sqrt{j(d_--j+1)}\Y_{k,\ell}^{j-1},\\
    0&= \sqrt{k(p-k+1)}\Y_{k-1,\ell}^j + \sqrt{\ell(d_+-\ell+1)}\Y_{k,\ell-1}^j - \sqrt{(j+1)(d_--j)}\Y_{k,\ell}^{j+1}.
\end{align*}
\end{enumerate}
Thus one only needs to solve the recurrence relations from (iii).
We note that (ii) implies that all three terms in the first (resp.\ second) equation vanish unless $p+d_+-d_-= 2(k+\ell-j+1)$ (resp.\ $p+d_+-d_-=2(k+\ell-j-1)$).
Assuming $p+d_+-d_-=0$ for simplicity (though not without loss of generality), the recurrence relations can be solved by induction on $k$, and the general solution is spanned by
\begin{equation}\label{eq-yukawa-matrix}
    \Y_{k,\ell}^j = 
    \begin{cases}
        \sqrt{\frac{\binom{p}{k}\binom{d_+}{\ell}}{\binom{d_-}{j}}}, & \text{if } j=k+\ell,\\
        0, & \text{otherwise}.
    \end{cases}
\end{equation}
The corresponding Yukawa map is then given by
\begin{equation*}
    \Y_w^+ v_\ell^+ = \lambda  \langle w_k, w \rangle \Y_{k,\ell}^j \, v_j^-,
    \qquad
    \Y_w^- v_j^- = -\conj{\lambda}\conj{\langle w_k, w\rangle} \Y_{k,\ell}^j \, v_\ell^+,
\end{equation*}
for a fixed $\lambda \in \C$, where summation over repeated indices is implied.
Since in our setting, $i\Y_w$ is self-adjoint, it is natural to set $\lambda = -im_\Y$ for a constant $m_\Y$, which is to be interpreted as the mass of the corresponding particles.

\begin{definition}\label{def-yukawa}
    Using the notation of Definition \ref{def-parameters}, assume that the representation parameters satisfy $q+h_+-h_-=0$ and $p+d_+-d_- = 0$,
    and assume $\Y_{k,\ell}^j$ is as in (\ref{eq-yukawa-matrix}).
    Then the map $w \mapsto \Y_w = \Y_w^- \oplus \Y_w^+$ given by
    \begin{equation*}
        \Y_w^+ v_\ell^+ = -im_\Y \langle w_k, w \rangle \Y_{k,\ell}^j \, v_j^-,
        \qquad
        \Y_w^- v_j^- = -im_\Y \conj{\langle w_k, w\rangle} \Y_{k,\ell}^j \, v_\ell^+
    \end{equation*}
    is said to be the \emph{Yukawa coupling of mass $m_\Y$}.
\end{definition}

\section{Static spherical symmetry}
\label{sec-static-spherical-symmetry}

\subsection{Symmetric principal bundles and associated objects}
\label{sec-symmetry-general}

Let us briefly recall the general notion of $K$-symmetry for manifolds, principal bundles, and associated objects. Here we only present a slightly simplified version of the definitions and classifications since they are sufficient for this paper. For the proofs and more details we refer to \cite{brodbeck, kunzle, harnad}, see also \cite{dym-adam-marko} for a somewhat more comprehensive review. 

Let $K$ and $G$ be compact Lie groups.
\begin{itemize}
    \item We say that a spacetime $(M, g)$ is $K$-symmetric if $K$ acts on $(M,g)$ by isometries. With a minor loss of generality, we can then assume that
    $M \cong N \times (K/H)$ for some closed subgroup $H \subset K$. We will also assume for simplicity of exposition that $N$ is contractible.
    \item A principal $G$-bundle over a $K$-symmetric manifold $M$ is said to be $K$-symmetric if the action of $K$ on $M$ lifts to an effective (left) action on $P$ by bundle automorphisms (i.e.\ it commutes with the right action of $G$). Such principal $G$-bundles are classified by conjugacy classes of homomorphisms $\lambda : H \to G$. In particular, the bundle $P_{\lambda}$ corresponding to a given homomorphism $\lambda$ can be constructed as the quotient
    \begin{equation}
    \label{eq-k-symmetric-pfb}
        P = N \times (K \times G) \mathbin{/} H, 
        \qquad (k,g) \sim (kh, \lambda(h^{-1})g),
    \end{equation}
    for $k \in K$, $g \in G$, $h \in H$, and where $P$ is
    equipped with the natural projection, right action by $G$, and left action by $K$.
    \item A connection $\omega$ on a $K$-symmetric principal bundle is said to be $K$-symmetric if it is invariant under the left action of $K$, i.e.\ $\ell_k^* \omega = \omega$ for $k\in K$. For the bundle $P_\lambda$ from above, the $K$-symmetric connections are classified by pairs $(\Lambda, \eta)$, where $\Lambda : N \to \text{Hom}(\k, \g)$ is a smooth map satisfying the Wang conditions
    \begin{equation*}
        \Lambda \circ \Ad_h = {\Ad_{\lambda(h)}} \circ \Lambda, \quad h \in H,
        \qquad\text{and}\qquad
        \Lambda|_{\mathfrak{h}} = \lambda_*,
    \end{equation*}
    (here $\k$, $\h$, and $\g$ are the Lie algebras of $K$, $H$, and $G$ respectively) and $\eta$ is a $\g$-valued one-form on $N$ invariant under the adjoint action of $\lambda(H) \subset G$.
    Explicitly, a choice of local section $\sigma : U\subset K/H \to K$ induces a local section
    \begin{equation*}
        \tilde\sigma : N \times U \subset M \to P, \qquad \tilde\sigma = (y, [\sigma, e]), \quad y \in N,
    \end{equation*}
    where $e\in G$ is the identity,
    under which the $K$-symmetric connection corresponding to $(\Lambda, \eta)$ pulls down to
    \begin{equation*}
        \tilde\sigma^*\omega = \Lambda \circ \sigma^\ast \mu_K + \eta,
    \end{equation*}
    where $\mu_K$ is the Maurer-Cartan form of $K$.
    \item If $\rho : G \to \mathbf{GL}(W)$ is a representation, then one can form the associated vector bundle $E = P\times_\rho W$. 
    Recall that sections $\Phi \in \Gamma(E)$ can be identified with maps $\phi : P \to W$ that are equivariant in the sense that $r_g^*\phi = \rho(g^{-1})\phi$, via $\Phi_{\pi(p)} = [p,\phi(p)]$.
    We say that a section $\Phi \in \Gamma(E)$ is $K$-symmetric if the associated map $\phi : P \to W$ is $K$-invariant in the sense that $\ell_k^* \phi = \phi$ for $k \in K$, or equivalently if the section $\Phi$ itself is (left) $K$-equivariant.
    For the bundle $P_\lambda$ from above, the $K$-symmetric sections $\Phi$ of the associated vector bundle $P \times_\rho W$ are classified by maps $N \to W$ satisfying the invariance condition
    \begin{equation}
    \label{eq-invariant-section-condition}
        \rho(\lambda(h)) \phi = \phi, \quad h \in H,
    \end{equation}
    or equivalently $\phi \in \mathrm{ker} (\rho_* \circ \lambda_*)|_{\h}$, provided that $H$ is connected.
    Explicitly, the associated symmetric section is given by
    \begin{equation*}
        \Phi_{(y, kH)} = [(y, [k, e]), \phi(y)], \qquad y \in N, \; k \in K,
    \end{equation*}
    where $e\in G$ is the identity.
\end{itemize}
This general theory is extremely useful for systematically determining the most general symmetric ansätze for all gauge theoretic fields, as we will see throughout this section. 
That being said, once the general ansatz is obtained, one can in principle forget about this theory and skip directly ahead to the calculations of the operators appearing in the equation system that one wishes to study, in our case (\ref{eq-yang-mills}--\ref{eq-einstein}).

\subsection{Spherical symmetry}

If the symmetry group $K = \SU(2)$ and $H = \U(1)$, the manifolds/bundles/sections are said to be \emph{spherically symmetric}.
In particular, the manifold then has topology
\begin{equation*}
    M = N \times (\SU(2) / \U(1)) = N \times \mathbb{S}^2.
\end{equation*}
In this context, we would like to recall the Hopf fibration
\begin{equation*}
    \pi_{\text{Hopf}} : \mathbb{S}^3 \cong \SU(2) \to \mathbb{S}^2 \cong \SU(2)/\U(1), \qquad (z,w) \mapsto (2\conj{z} w, |z|^2 - |w|^2)
\end{equation*}
where one identifies $\mathbb{S}^3 \subset \C^2$ with $\SU(2)$ via
\begin{equation*}
    (z,w) \leftrightarrow \begin{bmatrix}
        z & -\conj{w}\\
        w & \conj{z}
    \end{bmatrix},
    \qquad
    |z|^2+|w|^2 = 1.
\end{equation*}
For later purposes, it will be convenient to note that the map
\begin{equation}\label{eq-hopf-section}
    \sigma : (\theta,\varphi) \mapsto \exp(\varphi \tau_3) \exp(\theta \tau_2)
    = 
    \begin{bmatrix}
    e^{-\frac{i\varphi}{2}} \cos\frac{\theta}{2} & -e^{-\frac{i\varphi}{2}}\sin\frac{\theta}{2}  \\ 
    e^{\frac{i\varphi}{2}}\sin\frac{\theta}{2} & e^{\frac{i\varphi}{2}} \cos\frac{\theta}{2} 
    \end{bmatrix} 
\end{equation}
defines a local section of the bundle $\mathbb{S}^3 \to \mathbb{S}^2$ with respect to the standard spherical coordinates $(\theta,\varphi)$ defined on $\mathbb{S}^2$ minus the poles.
A calculation shows that
\begin{align}
    \nonumber
    \Ad_{\sigma(\theta,\varphi)} (\tau_1)
    &= \cos\theta\cos\varphi \, \tau_1 + \cos\theta\sin\varphi \, \tau_2 - \sin\theta \, \tau_3,\\
    \label{eq-Ad-cartesian-spherical}
    \Ad_{\sigma(\theta,\varphi)} (\tau_2)
    &= -\sin\varphi \,\tau_1 + \cos\varphi \,\tau_2,\\
    \nonumber
    \Ad_{\sigma(\theta,\varphi)} (\tau_3)
    &= \sin\theta\cos\varphi \, \tau_1 + \sin\theta\sin\varphi \, \tau_2 + \cos\theta \, \tau_3.
\end{align}
In particular, the image of $\sigma(\theta,\varphi)$ under the natural two-fold covering $\SU(2) \to \SO(3)$ (i.e.\ the adjoint map) provides the relation between the standard three-dimensional Cartesian basis and the spherical basis for $\su(2) \cong \R^3$.

\subsection{Spacetime and Einstein sector}

We consider a spacetime of the form 
\begin{equation*}
    M = \R \times I \times \mathbb{S}^2,
\end{equation*}
for some connected interval $I \subset \R$ (so that in the notation of the previous sections $N = I \times \R$).
Here, the first factor describes time  $t \in \R$ and the second factor describes a general radial coordinate $s \in I$.
Since we are working in the \emph{static} setting, our variables will generally depend only on $s$ and not on $t$.

We endow the manifold $M$ with a static spherically symmetric metric of the form
\begin{equation}
\label{eq-static-sph-sym-metric-general}
    g = -e^{2\alpha(s)} \diff t^2 + e^{2\beta(s)} \diff s^2 + r(s)^2 g_{\mathbb{S}^2},
\end{equation}
for some functions $\alpha, \beta, r : I \to \R$ where $r > 0$.
Note that only two of these three metric coefficients are truly independent, since one still has the freedom to choose the coordinate $s$.
There are several natural choices, but we defer this discussion to \S \ref{sec-coordinate-choices}.

We can define a natural (so-called spherical) local semi-orthonormal frame $\{\mathbf{e}_\mu\}$ for $M$ by setting
\begin{equation}\label{eq-frame-sph}
    \mathbf{e}_t = e^{-\alpha} \partial_t,
    \qquad
    \mathbf{e}_s = e^{-\beta} \partial_s,
    \qquad
    \mathbf{e}_\theta = \frac{1}{r} \partial_\theta,
    \qquad
    \mathbf{e}_\varphi = \frac{1}{r\sin\theta} \partial_\varphi,
\end{equation}
where $(\theta,\varphi)$ are the spherical coordinates, which are well-defined on $\mathbb{S}^2$ minus the poles, so the frame $\mathbf{e}_\mu$ is also defined on $\R \times I\times\mathbb{S}^2$ minus the poles. Its dual frame $\{\mathbf{e}^\mu\}$ is
\begin{equation}
    \mathbf{e}^t = e^{\alpha}\, \diff t,
    \qquad
    \mathbf{e}^s = e^{\beta} \,\diff s,
    \qquad
    \mathbf{e}^\theta = r\, \diff\theta,
    \qquad
    \mathbf{e}^\varphi = r\sin\theta\,\diff\varphi.
\end{equation}

The components of the Levi-Civita connection can then be computed as shown in Table \ref{tab:levi-civita}.
We can view the Levi-Civita connection as a connection $\omega^{\mathrm{LC}}$ on the frame bundle $\SO^+(M,g)$. 
If $\mathbf{e}$ is a section of $\SO^+(M,g)$, we have (cf.\ \cite[Beispiel 3.5]{baum})
\begin{equation*}
    (\mathbf{e}^*\omega^{\mathrm{LC}})(\mathbf{e}_\mu) = \sum_{\nu < \lambda} g(\nabla_{\mathbf{e}_\mu} \mathbf{e}_\nu, \mathbf{e}_\lambda) M^{\nu\lambda} = \frac12 g(\nabla_{\mathbf{e}_\mu} \mathbf{e}_\nu, \mathbf{e}_\lambda) M^{\nu\lambda},
\end{equation*}
where $M_{\mu\nu}$ are the matrices \eqref{eq-so13-basis}, indices are raised with respect to $\eta^{\mu\nu}$, and in the last equality summation over all indices is implied (hence the prefactor $1/2$).
Then with respect to the spherical frame (\ref{eq-frame-sph}), where the variables are indexed as%
\footnote{The ordering is not chosen randomly, and the naturality of this choice will become apparent later.}
$(0,1,2,3) \leftrightarrow (t,\theta,\varphi,s)$, we can write
\begin{equation}
\label{eq-levi-civita}
    \mathbf{e}^*\omega^{\mathrm{LC}} = -e^{-\beta} \alpha' \,\mathbf{e}^t \otimes M_{03} + e^{-\beta}\,\frac{r'}{r} ( \mathbf{e}^\theta \otimes M_{31} - \mathbf{e}^\varphi \otimes M_{23} ) + \frac{\cot\theta}{r} \, \mathbf{e}^\varphi \otimes M_{12}.
\end{equation}

\begin{table}[t]
    \centering
    \renewcommand{\arraystretch}{1.5}
    \begin{tabular}{|c||c|c|c|c|}
        \hline
        $\nabla_{\mathbf{e}_\mu} \mathbf{e}_\nu$  & $\nu=t$ & $\nu=s$ & $\nu=\theta$ & $\nu=\varphi$ \\\hline\hline
        $\mu=t$ & $e^{-\beta}\alpha'\, \mathbf{e}_s$ & $e^{-\beta}\alpha' \, \mathbf{e}_t$ & 0 & 0  
        \\\hline
        $\mu=s$ & 0 & 0 & 0 & 0  
        \\\hline
        $\mu=\theta$ & 0 & $e^{-\beta}\frac{r'}{r}\,\mathbf{e}_\theta$ & $- e^{-\beta}\frac{r'}{r}\, \mathbf{e}_s$ & 0  
        \\\hline
        $\mu=\varphi$ & 0  & $e^{-\beta}\frac{r'}{r} \, \mathbf{e}_\varphi$  & $\frac{\cot\theta}{r}\, \mathbf{e}_\varphi$ &  $- e^{-\beta}\frac{r'}{r}\, \mathbf{e}_s - \frac{\cot\theta}{r}\, \mathbf{e}_\theta$ \\\hline
    \end{tabular}
    \vskip0.25cm
    \caption{The covariant derivatives
    $\nabla_{\mathbf{e}_\mu} \mathbf{e}_\nu$ for the frame (\ref{eq-frame-sph}).}
    \label{tab:levi-civita}
\end{table}

Using Table \ref{tab:levi-civita}, one also computes the non-zero components of the Ricci tensor to be (cf.\ \cite[Exercise 14.16]{misner-thorne-wheeler})
\begin{align*}
    \Ric_g(\mathbf{e}_t, \mathbf{e}_t) 
    =&\, e^{-2\beta}\left[\alpha'' + \left( \alpha' - \beta' + \frac{2r'}{r} \right)\alpha'\right],
    \\[0.2cm]
    \Ric_g(\mathbf{e}_s, \mathbf{e}_s)
    =&\, - e^{-2\beta} \left[ \alpha'' + (\alpha' - \beta')\alpha' + \frac{2}{r}\left(r'' - \beta' r' \right) \right],
    \\[0.2cm]
    \Ric_g(\mathbf{e}_\theta, \mathbf{e}_\theta)
    = \Ric_g(\mathbf{e}_\varphi, \mathbf{e}_\varphi)
    =&\, \frac{1}{r^2}  -  e^{-2\beta} \left[ \frac{r''}{r} + \left( \alpha' - \beta' + \frac{r'}{r} \right)\frac{r'}{r} \right]
\end{align*} 
the scalar curvature is
\begin{align*}
    \Scal_g
    =&\, \frac{2}{r^2} - 2e^{-2\beta} \left[ \alpha'' + (\alpha' -\beta')\left(\alpha' + \frac{2r'}{r}\right) + \frac{2r''}{r} + \frac{r'{}^2}{r^2} \right],
\end{align*}
and the Einstein tensor $\Ein_g = \Ric_g - \frac12 \Scal_g g$ is
\begin{equation}
\label{eq-einstein-tensor-sph-sym}
\begin{split}
    \Ein_g 
    =&\, \left[ \frac{1}{r^2} - e^{-2\beta} \left( \frac{2r''}{r} + \left(\frac{r'}{r} - 2\beta'\right)\frac{r'}{r} \right) \right] \mathbf{e}^t \otimes \mathbf{e}^t
    \\[0.1cm]
    &+ \left[-\frac{1}{r^2} + e^{-2\beta}\left(2\alpha' + \frac{r'}{r}\right)\frac{r'}{r} \right] \mathbf{e}^s \otimes \mathbf{e}^s
    \\[0.1cm]
    & + e^{-2\beta} \left[ \alpha'' + \frac{r''}{r} + (\alpha' - \beta')\left(\alpha' + \frac{r'}{r}\right) \right]
    (\mathbf{e}^\theta \otimes \mathbf{e}^\theta + \mathbf{e}^\varphi \otimes \mathbf{e}^\varphi).
\end{split}
\end{equation}

\subsection{Principal bundles}
\label{sec-sph-sym-principal-bundles}

As hinted in \S \ref{sec-symmetry-general}, the spherically symmetric principal $G$-bundles are classified by homomorphisms $\lambda : \U(1) \to G$.
In our setting, we are dealing with two structure groups, namely the electroweak $G = \SU(2) \times \U(1)$ for the particles, and also the spin group $G = \Spin^+(1,3)$. We consider each of these separately.

\subsubsection{Particle bundle}
For the particle structure group $G = \SU(2) \times \U(1)$, the conjugacy classes of homomorphisms are given by
\begin{align*}
    \lambda = \lambda_{(n,m)} : \U(1) \subset \SU(2) &\to G = \SU(2) \times \U(1),
    \\[0.2cm]
    \begin{bmatrix}
        e^{it} & 0\\
        0 & e^{-it}
    \end{bmatrix}
    &\mapsto
    \left(
    \begin{bmatrix}
        e^{int} & 0\\
        0 & e^{-int}
    \end{bmatrix},
    e^{imt}
    \right)
\end{align*}
where $n,m\in\mathbb{Z}$ and without loss of generality $n\geq 0$.
Note that, with respect to the chosen basis matrices $\tau_k$ for $\su(2)\oplus\u(1)$ (cf.\ \S \ref{sec-particle-representations}), we can then write%
\footnote{The simplicity of this formula is one of the main reasons for choosing $\tau_4 = -i/2$, rather than just choosing it to just be $i$, since we have chosen the basis $\tau_\ell = -i\sigma_\ell/2$, $\ell=1,2,3$, as the natural basis for $\su(2)$.}
\begin{equation*}
    \lambda_{(n,m)} : \exp(t \tau_3) \mapsto \exp(nt\tau_3 + mt\tau_4) \cong (\exp(nt\tau_3), \, \exp(mt\tau_4)).
\end{equation*}
We will call $(n,m)$ the \emph{structure parameters} and the associated principal bundle from the classification by
\begin{equation*}
    P_{(n,m)} = N \times (\SU(2) \times G) \mathbin{/} \U(1),
\end{equation*}
where the quotient is taken with respect to $\lambda_{(n,m)}$ as described in (\ref{eq-k-symmetric-pfb}).

\subsubsection{Spin bundle} 
\label{sec-spin-bundle}
The spin group $G = \Spin^+(1,3)$ is not compact, so the aforementioned classification is not necessarily complete. Nevertheless, one can still obtain a spherically symmetric principal $\Spin^+(1,3)$-bundle via any homomorphism $\mu : \U(1) \to \Spin^+(1,3)$ by the same construction as in the classification. A natural choice here is obtained by using the standard identification $\SU(2) = \Spin(3)$ together with the standard embedding $\Phi : \SU(2) = \Spin(3)\to\Spin^+(1,3)$ described in \eqref{eq-SU2-Spin13-embedding}, i.e.\ we set 
\begin{align*}
    \mu : \U(1) \subset \SU(2) &\to \Spin^+(1,3)
    \\
    \exp(t\tau_3) &\mapsto \Phi(\exp(t\tau_3)).
\end{align*}
This provides a principal $\Spin^+(1,3)$-bundle $\Spin^+(M,g)$ defined by factoring in the same way as in (\ref{eq-k-symmetric-pfb}), i.e. we set
\begin{equation*}
    \Spin^+(M,g) = N \times (\SU(2) \times \Spin^+(1,3)) \mathbin{/} \U(1),
\end{equation*}
where the quotient is taken with respect to the homomorphism $\mu$.
We now only need to show that this is a spin structure, i.e.\ we need to construct an equivariant double covering $\mathscr{S} : \Spin^+(M,g) \to \SO^+(M,g)$ of the orthochronous frame bundle. To this end, we let $\rho : \Spin^+(1,3) \to \SO^+(1,3)$ be the standard group double covering, and we define a trivialization
\begin{equation*}
    E : \R^{1,3} \to T_{(t,s, kH)} M
\end{equation*}
given by (with respect to standard spherical coordinates $(\theta,\varphi)$)
\begin{align*}
    \gamma_0 &\mapsto \mathbf{e}_t = e^{-\alpha(s)}\partial_t
    \\
    \sin\theta\cos\varphi \, \gamma_1 + \sin\theta\sin\varphi \, \gamma_2 + \cos\theta \, \gamma_3 &\mapsto \mathbf{e}_s = e^{-\beta(s)} \partial_s,
    \\
    \cos\theta\cos\varphi \, \gamma_1 + \cos\theta\sin\varphi \, \gamma_2 - \sin\theta \, \gamma_3 &\mapsto \mathbf{e}_\theta = \frac{1}{r(s)} \partial_\theta,
    \\
    -\sin\varphi \, \gamma_1 + \cos\varphi \, \gamma_2 &\mapsto \mathbf{e}_\varphi = \frac{1}{r(s)\sin\theta} \partial_\varphi,
\end{align*}
where we identify $\gamma_\mu$ with the standard basis vectors of $\R^{1,3}$. Note that this map is globally defined (assuming $r>0$) and extends across the poles of the spherical coordinates. Then the desired double covering is given by
\begin{equation*}
    \mathscr{S} : [(t,s), k, g] \mapsto \left((t, s, kH), E_{(t, s, kH)} \circ \rho(\Phi(k)g) \right).
\end{equation*}
Note that if $\sigma(\theta, \varphi)$ is the standard Hopf section, then we get a natural spin frame 
\begin{equation*}
    \varepsilon(t,s, \theta,\varphi) = ((t,s), [\sigma(\theta,\varphi), e]),
\end{equation*}
as can easily be checked.
Using \eqref{eq-Ad-cartesian-spherical}, one can show that the associated semi-orthonormal (ordered) frame $\mathscr{S}\circ\varepsilon$ is given by the (ordered) spherical frame $(\mathbf{e}_t, \mathbf{e}_\theta, \mathbf{e}_\varphi, \mathbf{e}_s)$, which is also the reason for the ordering of the indices in \eqref{eq-levi-civita}.
In particular, with respect to this spin frame, a spinor field $\Psi \in \Gamma(\Sigma M)$ can be written (locally) as $\Psi = [\varepsilon, \psi]$ for a $\C^4$-valued function $\psi$, and Clifford multiplication operates as
\begin{equation*}
    \mathbf{e}_t \cdot \Psi = [\varepsilon, \, \gamma_0\psi], \quad
    \mathbf{e}_\theta \cdot \Psi = [\varepsilon, \, \gamma_1\psi], \quad
    \mathbf{e}_\varphi \cdot \Psi = [\varepsilon, \, \gamma_2\psi], \quad
    \mathbf{e}_s \cdot \Psi = [\varepsilon, \, \gamma_3\psi],
\end{equation*}
where we use the Weyl representation already described in \S \ref{sec-important-matrices}.
Furthermore, from (\ref{eq-double-cover-lie-algebra}, \ref{eq-levi-civita}), we see that the spin Levi-Civita connection is given in this frame by
\begin{align}
    \nonumber
    \epsilon^*\Gamma =&\, -e^{-\beta}\alpha' \, \mathbf{e}^t \otimes \frac12 \gamma_0\gamma_3 + \frac{\cot\theta}{r} \, \mathbf{e}^\varphi \otimes \frac12 \gamma_1\gamma_2
    \\
    \label{eq-spin-levi-civita}
    &\, + e^{-\beta}\,\frac{r'}{r} \left( \mathbf{e}^\theta \otimes \frac12\gamma_3\gamma_1 - \mathbf{e}^{\varphi} \otimes \frac12\gamma_2\gamma_3 \right).
\end{align}
We would like to note here that this spin structure coincides with the standard spin structure one uses on spherically symmetric spacetimes, albeit this is rarely described as explicitly as we do here.

\subsection{Yang-Mills sector}
\label{sec-yang-mills-ansatz}

Let us now derive the most general ansatz for the connection on $P_{(n,m)}$.
This ansatz is well-known in literature, in particular for $(n,m)=(1,0)$ it is more commonly known as the \emph{Witten ansatz}.
Nevertheless, we include some details of the derivation for completeness. 

\subsubsection{Derivation of the ansatz}
The symmetric connections on $P_{(n,m)}$ are classified by $(\Lambda,\eta)$ as described in \S \ref{sec-symmetry-general}.
Using the Wang conditions and the fact that $\tau_4$ belongs to the center of $\mathfrak{su}(2)\oplus\mathfrak{u}(1)$, one easily shows that
\begin{equation*}
    \Lambda(\tau_1) = -n[\tau_3, \Lambda(\tau_2)],\qquad
    \Lambda(\tau_2) = n[\tau_3, \Lambda(\tau_1)],\qquad
    \Lambda(\tau_3) = n\tau_3 + m\tau_4.
\end{equation*}
Writing out $\Lambda(\tau_1), \Lambda(\tau_2)$ in the basis $\tau_i$, the first two identities imply that 
\begin{equation*}
    \Lambda(\tau_1) = 
    \begin{cases}
        0, & n\not=1,\\
        v\tau_1 - w\tau_2, & n=1
    \end{cases},
    \qquad
    \Lambda(\tau_2) = 
    \begin{cases}
        0, & n\not=1,\\
        w\tau_1 + v\tau_2, & n=1
    \end{cases},
\end{equation*}
for real $v,w$.
On the other hand, for $\eta$, we note that $\tau_3$ and $\tau_4$ span the invariant subspace of the adjoint action of $\lambda_{(n,m)}(\U(1))$, so that one has
\begin{equation*}
    \eta = (a \, \diff t + c  \, \diff s) \otimes \tau_3 + (b \, \diff t + d\, \diff s) \otimes \tau_4,
\end{equation*}
for any real $a,b,c,d$.
The pair $(\Lambda,\eta)$ then defines a symmetric connection on $P_{(n,m)}$ for any choice of coefficients $a,b,c,d,v,w$ depending on $(t,s) \in \R \times I = N$.
To write this connection down locally, we consider the section $\sigma$ of the Hopf bundle described in (\ref{eq-hopf-section}).
One calculates
\begin{equation*}
    \sigma^\ast \mu_{\SU(2)} = \sigma^{-1}\diff\sigma 
    = \diff\theta \otimes \tau_2 - \sin\theta\, \diff\varphi \otimes \tau_1 + \cos\theta\, \diff\varphi \otimes \tau_3,
\end{equation*}
and therefore
\begin{align*}
    \Lambda \circ \sigma^\ast \mu_{\SU(2)}
    =&\,
    v(\diff\theta \otimes \tau_2 - \sin\theta \, \diff\varphi \otimes \tau_1)
    + w(\diff\theta \otimes \tau_1 + \sin\theta \, \diff\varphi \otimes \tau_2)\\[0.1cm]
    &+ \cos\theta \, \diff\varphi \otimes (n\tau_3 + m\tau_4).
\end{align*}
It follows that, with respect to the induced section $\tilde{\sigma} = ((t,s), [\sigma, e])$, the most general static spherically symmetric connection on $P_{(n,m)}$ can be written as
\begin{align*}
    \tilde{\sigma}^*\omega =&\, (a \, \diff t + c  \, \diff s) \otimes \tau_3 + (b \, \diff t + d\, \diff s) \otimes \tau_4
    + \cos\theta\, \diff\varphi \otimes (n\tau_3 + m\tau_4)
    \\[0.1cm]
    &+ w(\diff\theta \otimes \tau_1 + \sin\theta\,\diff\varphi \otimes \tau_2) + v (\diff\theta \otimes \tau_2 - \sin\theta\,\diff\varphi \otimes \tau_1),
\end{align*}
where $w,v = 0$ if $n\not=1$.

One can without loss of generality either make $c=d=0$ (resp.\ $a=b=0$) by a gauge transformation that is $s$-dependent (resp.\ $t$-dependent). Since only $s$-dependent gauge transformations are compatible with the staticity assumption, it is common at this stage to take $c = d =0$, and we will do so as well.

It is also common in pure $\SU(2)$ (Einstein-)Yang-Mills theory to assume that $v = 0$, which can be done without loss of generality in such simplified theories (i.e.\ without other matter fields than the Yang-Mills field) since the pure Yang-Mills equation for this ansatz implies that $v$ and $w$ must be constant multiples of one another, and then a constant gauge transformation takes $v \mapsto 0$ (see also \S \ref{sec-further-reductions}). However, when more fields are involved as in the setting of this paper, then the Yang-Mills equation instead just provides a certain constraint between $v$ and $w$ (cf.\ \S \ref{sec-constraints-invariances}) and in particular one seemingly cannot set $v = 0$ without losing generality.

\subsubsection{Operator formulas}
Since we do most of our computations with respect to frames, rather than coordinates, it is convenient to write the connection form as
\begin{align}
    \nonumber
    \tilde{\sigma}^*\omega =&\, \mathbf{e}^t \otimes e^{-\alpha}(a\tau_3 + b\tau_4)
    + \mathbf{e}^\varphi \otimes \frac{\cot\theta}{r}\,(n\tau_3 + m\tau_4)
    \\
    \label{eq-sph-sym-connection}
    &+ \frac{\Re(z)}{r} \,(\mathbf{e}^\theta \otimes \tau_1 + \mathbf{e}^\varphi \otimes \tau_2) + \frac{\Im(z)}{r} \,(\mathbf{e}^\theta \otimes \tau_2 - \mathbf{e}^\varphi \otimes \tau_1),
\end{align}
where we have also defined
\begin{equation*}
    z = \begin{cases}
        w+iv, & \text{if } n = 1,\\
        0, & \text{otherwise}.
    \end{cases}
\end{equation*}
Furthermore, we are interested only in the \emph{static} Euler-Lagrange equations, and consequently we will assume that the coefficients $a,b,c,d,v,w$ depend only on $s$.

Let us describe the curvature of $\omega$ and associated operators that are relevant for studying the Euler-Lagrange equations.
All of the formulas are expressed with respect to the section $\tilde{\sigma}$ as above, which will be suppressed for notational simplicity.
The curvature of this connection is given by
\begin{align*}
    F_\omega =&\;
    -\mathbf{e}^t \wedge \mathbf{e}^s \otimes e^{-\alpha-\beta}(a' \tau_3 + b' \tau_4)
    - \mathbf{e}^\theta \wedge \mathbf{e}^\varphi \otimes \frac{1}{r^2}\left((n-|z|^2) \tau_3 + m\tau_4\right)
    \\[0.1cm]
    &+ \frac{1}{r} (e^{-\alpha}a\Re(iz)\, \mathbf{e}^t + e^{-\beta}\Re(z')\, \mathbf{e}^s) \wedge (\mathbf{e}^\theta \otimes \tau_1 + \mathbf{e}^\varphi \otimes \tau_2)
    \\[0.1cm]
    &+ \frac{1}{r} (e^{-\alpha}a\Im(iz)\, \mathbf{e}^t + e^{-\beta}\Im(z')\, \mathbf{e}^s) \wedge (\mathbf{e}^\theta \otimes \tau_2 - \mathbf{e}^\varphi \otimes \tau_1).
\end{align*}
The Yang-Mills operator satisfies
\begin{equation}
\label{eq-ym-operator}
\begin{split}
    \diff_\omega^* F_\omega =& - \frac{1}{r^2}\left[e^{-\beta} (r^2 e^{-\alpha-\beta}a')' - 2 e^{-\alpha} a|z|^2 \right] \mathbf{e}^t \otimes \tau_3
    \\
    &+ \frac{2}{r^2} e^{-\beta}\Im(\conj{z}z') \, \mathbf{e}^s \otimes \tau_3
    - \frac{1}{r^2} e^{-\beta} (r^2 e^{-\alpha-\beta}b')' \, \mathbf{e}^t \otimes \tau_4
    \\[0.1cm]
    &- \frac{1}{r} \, \Re \left[e^{-\alpha-\beta}(e^{\alpha-\beta}z')' + e^{-2\alpha}a^2z + \frac{1}{r^2}z(n-|z|^2) \right] (\mathbf{e}^\theta \otimes \tau_1 + \mathbf{e}^\varphi \otimes \tau_2)
    \\[0.1cm]
    &- \frac{1}{r} \, \Im \left[e^{-\alpha-\beta}(e^{\alpha-\beta}z')' + e^{-2\alpha}a^2z + \frac{1}{r^2}z(n-|z|^2) \right] (\mathbf{e}^\theta \otimes \tau_2 - \mathbf{e}^\varphi \otimes \tau_1).
\end{split}
\end{equation}
The energy-momentum tensor is given by
\begin{align*}
    \mathfrak{T}^{\mathrm{YM}} 
    =&\, \langle \cdot \iprod F_\omega, \cdot \iprod F_\omega \rangle - \frac12 |F_\omega|^2 g
    \\[0.1cm]
    =& \mathrel{\phantom{+}} \bigg[
    \frac12 e^{-2\alpha-2\beta} (a'{}^2 + b'{}^2)
    + \frac{1}{r^2}\left (e^{-2\alpha}a^2|z|^2 + e^{-2\beta}|z'|^2 
    + \frac{(n-|z|^2)^2+m^2}{2r^2} \right)
    \bigg]\, \mathbf{e}^t \otimes \mathbf{e}^t
    \\[0.1cm]
    &+ \frac{2}{r^2} e^{-\alpha-\beta} a\Im(\conj{z}z')  (\mathbf{e}^t \otimes \mathbf{e}^s + \mathbf{e}^s \otimes \mathbf{e}^t)
    \\[0.1cm]
    &- \bigg[ \frac12 e^{-2\alpha-2\beta} (a'{}^2 + b'{}^2)
    - \frac{1}{r^2} \left( e^{-2\alpha}a^2|z|^2 + e^{-2\beta}|z'|^2 - \frac{(n-|z|^2)^2+m^2}{2r^2} \right) \bigg] \, \mathbf{e}^s \otimes \mathbf{e}^s
    \\[0.1cm]
    &+ \frac12 \left[ e^{-2\alpha-2\beta} (a'{}^2 + b'{}^2)
    + \frac{(n-|z|^2)^2+m^2}{r^4} \right] (\mathbf{e}^\theta \otimes \mathbf{e}^\theta + \mathbf{e}^\varphi \otimes \mathbf{e}^\varphi).
\end{align*}

\subsection{Higgs sector}
\label{sec-higgs-ansatz}

Next, we study the Higgs sector and general ansatz for the Higgs field.

\subsubsection{Derivation of the ansatz}
We recall that the Higgs field is a section of an associated vector bundle $P_{(n,m)} \times_{\rho_{(p,q)}} W$ for some choice of representation parameters $(p,q)$, cf.\ \S \ref{sec-particle-representations}.
Note that we can equivalently state the invariance condition (\ref{eq-invariant-section-condition}) as $\phi \in \ker \rho_*(\lambda_*(\tau_3))$.
We then note that
\begin{equation*}
    (\rho_{(p,q)})_*(\lambda_*(\tau_3))(w_k) 
    = (\rho_{(p,q)})_*(n\tau_3 + m\tau_4)(w_k) 
    = -\frac{i}{2}(n(p-2k)+qm)w_k.
\end{equation*}
Since this action is diagonal, it follows that invariant vectors of this representation are spanned by $w_k$ whose index $k$ satisfies
\begin{equation*}
    n(p-2k)+qm = 0.   
\end{equation*} 
\begin{itemize}[itemsep=0.1cm]
    \item If $n=0$, then necessarily $qm=0$, and the choice of $p$ and $k$ is immaterial.
    \item If $n > 0$, then
    \begin{equation*}
        k = \frac12\left(p+\frac{qm}{n}\right) \in \{0,1,\ldots, p\}.
    \end{equation*}
    Then it is necessary than $qm/n \in \Z$, $|qm|/n \leq p$, and $qm/n$ has the same parity as $p$.
\end{itemize}

In this case, we obtain a static%
\footnote{As before, the coefficient $\phi$ can also depend on $t$ if one wants the most general spherically symmetric ansatz (i.e.\ not necessarily static), but we are here specializing to the static setting and we will continue doing so without explicit mention.}
spherically symmetric section $\Phi$ of $P \times_{\rho_{(p,q)}} W$ by setting
\begin{equation*}
    \Phi = [((t,s), [k_0, e]),\; \phi(s)w_k], \qquad k_0 \in K=\SU(2).
\end{equation*}

\subsubsection{Operator formulas}
We now calculate $\nabla_\omega\Phi$ with respect to a spherically symmetric connection, and with respect to the section $\tilde{\sigma} = ((t,s), [\sigma(\theta,\varphi), e])$ of $P$.
To simplify notation, we define the real-valued function
\begin{equation}\label{eq-zeta-def}
    \zeta^{(p,q)}_j = -\frac{1}{2}\left[ (p-2j)a + qb \right],
\end{equation}
so that the action of the connection \eqref{eq-sph-sym-connection} can be written as (no summation over repeated $j$)
\begin{align}
    (\rho_{(p,q)})_*(\tilde{\sigma}^*\omega)w_j =& \;
    ie^{-\alpha}\zeta^{(p,q)}_j\, \mathbf{e}^{t} \otimes w_j - \frac{i}{2} (n(p-2j)+mq)\,\frac{\cot\theta}{r}\,\mathbf{e}^\varphi\otimes w_j
    \nonumber
    \\
    &-\frac{i}{2}\sqrt{(j+1)(p-j)}\, \frac{z}{r}\, (\mathbf{e}^\theta + i\mathbf{e}^\varphi) \otimes w_{j+1}
    \label{eq-sph-sym-connection-associated-action}
    \\
    \nonumber
    &-\frac{i}{2}\sqrt{j(p-j+1)}\, \frac{\conj{z}}{r}\, (\mathbf{e}^\theta - i\mathbf{e}^\varphi) \otimes w_{j-1}.
\end{align}
Note that the invariance condition ensures precisely that the $\cot\theta$ term vanishes%
\footnote{In fact, we would like to note at this stage that one could in principle also forget about the invariance conditions from \S \ref{sec-symmetry-general} and just calculate the derivatives on the entire space $W$ and then isolate the components that are spherically symmetric in the sense that the formulas do not depend on the angles.
The benefit of using the invariance condition however is that it identifies such components directly and in a systematic manner, which becomes convenient when working with more complicated spaces (such as the twisted spinors later).},
and at the invariant index $j=k$ (for $n > 1$) we could also simplify
\begin{equation*}
    \zeta^{(p,q)}_k = \frac{q}{2}\left(\frac{m}{n} a - b\right).
\end{equation*}
Then on static spherically symmetric $\Phi$, we have
\begin{align}
    \nonumber
    \nabla_\omega\Phi
    =&\,  \mathbf{e}^t\otimes ie^{-\alpha}\zeta_k^{(p,q)}\phi \, w_k
    + \mathbf{e}^s \otimes e^{-\beta} \phi'\, w_k
    \\\label{eq-higgs-covariant-derivative}
    &+\mathbf{e}^\theta \otimes \frac{i\phi}{2 r} \left(-\conj{z} \sqrt{k(p-k+1)} \, w_{k-1} - z \sqrt{(k+1)(p-k)}\, w_{k+1}\right)
    \\\nonumber
    &+ \mathbf{e}^\varphi \otimes \frac{\phi}{2 r} \left(-\conj{z} \sqrt{k(p-k+1)} \, w_{k-1} + z \sqrt{(k+1)(p-k)}\, w_{k+1}\right).
\end{align}
The wave operator is%
\footnote{Note that for $n = 1$, the $\cot \theta$ terms in the wave operator cancel, while for $n \not= 1$ we have $z = 0$ and the $\cot \theta$ terms are not present at all.}
\begin{equation*}
    \Box_\omega \Phi
    = \left\{
    e^{-2\beta} \bigg[ \phi'' + \left(\alpha' - \beta' + \frac{2r'}{r} \right)\phi' \bigg]
    + \left[e^{-2\alpha}  (\zeta^{(p,q)}_k)^2  
    - \frac{\kappa|z|^2}{2r^2} \right] \phi \right\} w_k,
\end{equation*}
where we denote
\begin{equation}
\label{eq-kappa}
    \kappa = k(p-k+1)+(k+1)(p-k) = \frac12\left(p(p+2)-\frac{q^2m^2}{n^2}\right).
\end{equation}
The energy-momentum tensor is
\begin{align*}
    \mathfrak{T}^{\mathrm{H}} =&\,
    \Re\langle\nabla_\omega\Phi \otimes \nabla_\omega\Phi\rangle - \frac12 (|\nabla_\omega\Phi|^2 + U(\Phi)) g
    \\[0.1cm]
    =&\mathrel{\phantom{+}} \frac12\left[ e^{-2\alpha}|\zeta_k^{(p,q)}\phi|^2 + e^{-2\beta}|\phi'|^2 + \frac{\kappa|z|^2|\phi|^2}{2r^2} + U(\phi)\right] \mathbf{e}^t \otimes \mathbf{e}^t
    \\[0.1cm]
    &+ \frac12\left[ e^{-2\alpha}|\zeta_k^{(p,q)}\phi|^2 + e^{-2\beta}|\phi'|^2 - \frac{\kappa|z|^2|\phi|^2}{2r^2} - U(\phi)\right] \mathbf{e}^s \otimes \mathbf{e}^s
    \\[0.1cm]
    &+ \zeta_k^{(p,q)} e^{-\alpha-\beta}\Im(\conj{\phi}\phi') (\mathbf{e}^t \otimes \mathbf{e}^s + \mathbf{e}^s \otimes \mathbf{e}^t)
    \\[0.1cm]
    &+ \frac12 \left[ e^{-2\alpha}|\zeta_k^{(p,q)}\phi|^2 - e^{-2\beta}|\phi'|^2 - U(\phi) \right](\mathbf{e}^\theta\otimes\mathbf{e}^\theta + \mathbf{e}^\varphi \otimes \mathbf{e}^\varphi),
\end{align*} 
where $U$ is the Higgs potential \eqref{eq-mexican-hat-potential}.

\subsection{Dirac sector}
\label{sec-dirac-ansatz}

Finally, we study the Dirac sector, which is modelled by twisted spinors.
It is well-known (and we shall also see this below) that there exist no spherically symmetric pure spinors, cf.\ \cite[Example 1.2]{sph-sym-standard-model}.
However, when spinors are twisted by an additional coefficient bundle, then it becomes possible to produce a spherically symmetric ansatz for the field in certain cases, as we shall see below.

\subsubsection{Derivation of the ansatz}

Twisted spinors are sections of $\Sigma M \otimes \mathscr{S}$, where $\mathscr{S} = P \times_\chi V$ is a vector bundle associated to $P$ via a representation $\chi : G \to \U(V)$. 
Note that the spinor bundle itself is also an associated vector bundle $\Sigma M = \Spin^+(M,g) \times_\kappa \C^4$ where $\kappa$ is the spinor representation, i.e.\ the restriction of the Clifford representation to $\Spin^+(1,3)$.
Hence we can view
\begin{equation*}
    \Sigma M \otimes \mathscr{S} = (\Spin^+(M,g) \times_\kappa \C^4) \otimes (P \times_\chi V) \cong (\Spin^+(M,g) \times P) \times_{\kappa \otimes \chi} (\C^4 \otimes V)
\end{equation*}
as an associated vector bundle to the principal fiber bundle $\Spin^+(M,g) \times P \to M$ with structure group $\Spin^+(1,3) \times G$ with respect to the representation
\begin{equation*}
    \kappa \otimes \chi : \Spin^+(1,3) \times G \to \mathbf{GL}(\C^4 \otimes V), \qquad (b, g) \mapsto ( u \otimes v \mapsto \kappa(b)u \otimes \chi(g)v). 
\end{equation*}
If the spin bundle $\Spin^+(M,g)$ and $P$ are $K$-symmetric, then trivially so is their fibered product $\Spin^+(M,g) \times P \to M$.
In particular, if $\mu : \U(1) \to \Spin^+(1,3)$ is the classifying homomorphism for $\Spin^+(M, g)$ and $\lambda = \lambda_{(n,m)} : \U(1) \to \SU(2) \times \U(1)$ is the classifying homomorphism for $P$, then the classifying homomorphism for the fibered product is
\begin{equation*}
    \tilde{\lambda} = \mu \times \lambda_{(n,m)} : \U(1) \to \Spin^+(1,3) \times G, \qquad \tilde{\lambda}(h) = (\mu(h), \lambda_{(n,m)}(h)).
\end{equation*}
Thus the general conditions for symmetric sections of vector bundles apply here as well, namely invariant twisted spinors are induced by elements of $\C^4\otimes V$ that reside in the kernel of 
\begin{align*}
    ((\kappa\otimes\chi) \circ \tilde{\lambda})_*(\tau_3) 
    &= ((\kappa\circ \mu) \otimes (\chi \circ \lambda_{(n,m)}))_*(\tau_3)\\
    &= 
    \kappa_*(\mu_*(\tau_3)) \otimes \id + \id \otimes \chi_*((\lambda_{(n,m)})_*(\tau_3)) \in \mathfrak{gl}(\C^4 \otimes V).
\end{align*}
Now $\kappa$ is the spinor representation and by \eqref{eq-U(1)-inside-spin} we have $\kappa(\mu(\exp(s\tau_3))) = I_2 \otimes \exp(s\tau_3)$
so that infinitesimally
\begin{equation*}
    \kappa_*(\mu_*(\tau_3)) \u_\pm = -\frac{i}{2}\u_\pm,
    \qquad
    \kappa_*(\mu_*(\tau_3)) \d_\pm = +\frac{i}{2}\d_\pm,
\end{equation*}
where we use the standard basis \eqref{eq-spinor-up-down-chiral-basis} for the spinor space in the Weyl representation.
This already shows that one cannot have a pure (i.e.\ non-twisted) spinor field satisfying the invariance condition, i.e.\ there exist no pure spinors in spherical symmetry, cf.\ \cite[Example 1.2]{sph-sym-standard-model}.
On the other hand, when the spinor is twisted then the twisting action can "cancel out" the actions above to satisfy the invariance condition.
We take $\chi = \rho_{(d,h)}$ to be the $(d,h)$-representation with basis $v_\ell$, and we denote the associated vector bundle by $\mathscr{S}_{(n,m,d,h)} = P_{(n,m)} \times_{\rho_{(d,h)}} V$.
It follows then that
\begin{align*}
    \left((\kappa\circ\mu) \otimes (\rho_{(d,h)}\circ\lambda_{(n,m)})\right)_*(\tau_3)(\u_\pm \otimes v_\ell) &= -\frac{i}{2}[n(d-2\ell)+hm+1](\u_\pm \otimes v_\ell)
    \\
    \left((\kappa\circ\mu) \otimes (\rho_{(d,h)}\circ\lambda_{(n,m)})\right)_*(\tau_3)(\d_\pm \otimes v_\ell) &= -\frac{i}{2}[n(d-2\ell)+hm-1](\d_\pm \otimes v_\ell).
\end{align*}
Thus, the subspace of invariant twisted spinors is generated by all $v_\ell$ whose index $\ell$ satisfies 
\begin{equation*}
    n(d-2\ell)+hm\pm 1 = 0,   
\end{equation*}
which certainly has solutions provided that the parameters are chosen correctly.
More precisely:
\begin{itemize}[itemsep=0.1cm]
\item If $n=0$, then necessarily $hm = \pm 1$. 
In this case the choice of $d$ (and hence $\ell$) is immaterial.
With no major loss of generality we can also set $m = 1$ and $h = \pm 1$.
Then we can set
\begin{equation*}
    \psi =
    \begin{cases}
        (\xi_+ \mathfrak{u}_{+} + \xi_- \mathfrak{u}_{-}) \otimes v_\ell, & \text{if }  h = -1\\
        (\eta_+ \mathfrak{d}_{+} + \eta_- \mathfrak{d}_{-}) \otimes v_\ell, & \text{if } h = 1\\
        0, & \text{otherwise},
    \end{cases}
\end{equation*}
for any $\ell$ and complex $\xi_\pm, \eta_\pm$.
\item If $n=1$, then we set
\begin{equation*}
    \ell=\frac12\left(d+hm+1\right),
    \qquad
    \ell-1 = \frac12\left(d+hm-1\right)
\end{equation*}
and invariant twisted spinors exist provided that $\ell \in \{ 0, \ldots, d+1 \}$. Note that $\ell = 0$ (resp.\ $\ell=d+1$) are degenerate cases, in the sense that $\ell-1$ is no longer non-negative (resp.\ $\ell > d$ but $\ell-1 = d$).
Hence, we define
\begin{equation*}
    \psi =
    \begin{cases}
        (\xi_+ \mathfrak{u}_{+} + \xi_- \mathfrak{u}_{-}) \otimes v_\ell + (\eta_+ \mathfrak{d}_{+} + \eta_- \mathfrak{d}_{-}) \otimes v_{\ell-1}, & \text{if } \ell \in \{1,\ldots,d\}\\
        (\xi_+ \mathfrak{u}_{+} + \xi_- \mathfrak{u}_{-}) \otimes v_\ell, & \text{if } \ell = 0\\
        (\eta_+ \mathfrak{d}_{+} + \eta_- \mathfrak{d}_{-}) \otimes v_{\ell-1}, & \text{if } \ell-1 = d\\
        0, & \text{otherwise},
    \end{cases}
\end{equation*}
for any complex $\xi_\pm, \eta_\pm$.
\item If $n>1$, then $\frac12\left(d+\frac{hm\pm1}{n}\right)$ is an integer for at most one choice of sign.
Here we can set
\begin{equation*}
    \psi =
    \begin{cases}
        (\xi_+ \mathfrak{u}_{+} + \xi_- \mathfrak{u}_{-}) \otimes v_\ell, & \text{if } \ell = \frac12\left(d+\frac{hm + 1}{n}\right) \in \{0,\ldots,d\}\\
        (\eta_+ \mathfrak{d}_{+} + \eta_- \mathfrak{d}_{-}) \otimes v_{\ell-1}, & \text{if } \ell-1 = \frac12\left(d+\frac{hm - 1}{n}\right) \in \{0,\ldots,d\}\\
        0, & \text{otherwise},
    \end{cases}
\end{equation*}
for any complex $\xi_\pm, \eta_\pm$.
\end{itemize}
In particular, we see that the case $n=1$ with $\ell \in \{1,\ldots,d\}$ provides the richest ansatz, just like in the Yang-Mills sector.
In all of these cases, the general static spherically symmetric twisted spinor is then defined as the section of $\Sigma M \otimes \mathscr{S}_{(n,m,d,h)}$ given by
\begin{equation*}
    \Psi_{((t,s), kH)} = [((t,s), [k, e]), \; \psi(s)],
\end{equation*}
where the coefficients $\xi_\pm, \eta_\pm$ also depend only on $s$.
If one demands that that the twisted spinors also be chiral, then one  additionally needs to assume that either $\xi_+ = \eta_+ = 0$ or $\xi_- = \eta_- = 0$. At this stage we shall not impose any such conditions, in order to get the most general ansatz, but in later sections we will work with twisted chiral spinors $\Psi = \Psi^+ + \Psi^-$, where $\Psi^\pm$ are chiral and both have the form as above, each for a separate choice of twisting representation.

\subsubsection{Operator formulas}
\label{sec-dirac-operator-formulas}

The connection on $\Sigma M \otimes E_{(n,m,d,h)}$ is the connection associated to $\Gamma \oplus \omega$, where $\Gamma$ is the spin Levi-Civita connection on $\Spin^+(M,g)$ \eqref{eq-spin-levi-civita}, and $\omega$ is a static spherically symmetric connection on $P$ \eqref{eq-sph-sym-connection}. Then, with respect to the spin frame $\varepsilon$ as defined in \S \ref{sec-spin-bundle} and the lifted section $\tilde{\sigma} = ((t,s), [\sigma(\theta,\varphi), e])$ of $P$ where $\sigma$ is the Hopf section (\ref{eq-hopf-section}), we have $\nabla_\omega\Psi = [\varepsilon\times \tilde{\sigma}, \nabla_\omega\psi]$, where
\begin{equation*}
    \nabla_\omega \psi = \diff\psi + \kappa_*(\varepsilon^*\Gamma)\psi + \chi_*(\tilde{\sigma}^*\omega)\psi,
\end{equation*}
and the twisted Dirac operator is given by
\begin{equation*}
    \dirac_\omega \Psi = \eta^{\mu\nu} \mathbf{e}_\mu \cdot (\nabla_\omega \Psi)(\mathbf{e}_\nu) = [\varepsilon \times \tilde\sigma, \, \gamma^\mu (\nabla_\omega\psi)_\mu].
\end{equation*}
One computes the non-vanishing components of the Dirac operator to be (with respect to the section $\varepsilon \times \tilde\sigma$)
\begin{align*}
    (\dirac_\omega\Psi)_{\u_\pm \otimes v_\ell} 
    =& \,
    \pm e^{-\beta} \left[\xi_\mp' + \left(\frac12\alpha' + \frac{r'}{r}\right)\xi_\mp \right]
    - ie^{-\alpha}\zeta^{(d,h)}_\ell \xi_\mp \mp i\sqrt{\ell(d-\ell+1)}\,\frac{z\eta_\mp}{r},
    \\[0.2cm]
    (\dirac_\omega\Psi)_{\d_\pm \otimes v_{\ell-1}} 
    =& \,
    \mp e^{-\beta} \left[\eta_{\mp}' + \left(\frac12 \alpha' + \frac{r'}{r}\right)\eta_{\mp}\right] 
    - ie^{-\alpha} \zeta^{(d,h)}_{\ell-1}\eta_{\mp} \mp i\sqrt{\ell(d-\ell+1)}\,\frac{\conj{z}\xi_{\mp}}{r},
\end{align*}
where $\zeta$ is defined by the same formula as in \eqref{eq-zeta-def}.
Here, the subscripts indicate the basis elements that the coefficient on the right-hand side corresponds to, 
and it is understood that the coefficients $\xi_{\pm}$ and $\eta_{\pm}$ vanish in certain cases, in line with the different choices of parameters, as shown above, in particular when $\ell$ or $\ell-1$ are not integers in $\{0,\ldots,d\}$, and we also recall again that $z$ vanishes unless $n=1$.
The energy-momentum tensor is (assuming that the Dirac-Yukawa equation $\dirac_\omega \Psi + \Y_\Phi\Psi = 0$ is satisfied)
\begin{align*}
    \mathfrak{T}^{\mathrm{D}}
    =&
    -\frac12 e^{-\alpha}\bigg[ \zeta^{(d,h)}_\ell (|\xi_+|^2+|\xi_-|^2)
    + \zeta^{(d,h)}_{\ell-1} (|\eta_+|^2+|\eta_-|^2)\bigg] \, \mathbf{e}^t \otimes \mathbf{e}^t
    \\[0.1cm]
    &+ \frac12 e^{-\alpha}\bigg[ \zeta^{(d,h)}_\ell (|\xi_+|^2-|\xi_-|^2) 
    - \zeta^{(d,h)}_{\ell-1} (|\eta_+|^2-|\eta_-|^2)
    \bigg] (\mathbf{e}^t \otimes \mathbf{e}^s + \mathbf{e}^s \otimes \mathbf{e}^t)
    \\[0.1cm]
    &+ \frac12 e^{-\beta} \, \Im\left(\conj{\xi}_+ \xi_+' - \conj{\xi}_- \xi_-' - \conj{\eta}_+ \eta_+' + \conj{\eta}_- \eta_-' \right) \mathbf{e}^s \otimes \mathbf{e}^s
    \\[0.1cm]
    &+ 
    \frac{\sqrt{\ell(d-\ell+1)}}{2r} \left[ -\Re(z\conj{\xi}_+\eta_+) + \Re(z\conj{\xi}_-\eta_-) \right]
    (\mathbf{e}^\theta \otimes \mathbf{e}^\theta + \mathbf{e}^\varphi \otimes \mathbf{e}^\varphi).
\end{align*}

\subsubsection{Chirality}
\label{sec-chiral}
Finally, having written down these general formulae, let us briefly also mention how to use them when deriving the formulae for the twisted chiral spinors.

To this end, assume $\Psi^{\pm} \in \Gamma(\Sigma_\pm \otimes \mathscr{S}_\pm)$ are symmetric twisted spinors as above, where
$\mathscr{S}_\pm = \mathscr{S}_{(n,m,d_\pm,h_\pm)}$ for some choices of representation parameters as above (so that $\chi_\pm = \rho_{(d_\pm, h_\pm)}$ in the notation of \S \ref{sec-sm-lagrangian}).
Then $\Psi=\Psi^++\Psi^-$ forms a symmetric twisted chiral spinor, i.e.\ a section of $\mathscr{F}_+$, cf.\ \eqref{eq-twisted-chiral-bundle}.
We note that
\begin{equation*}
    \dirac_\omega \Psi = \dirac_\omega \Psi^+ + \dirac_\omega \Psi^-,
\end{equation*}
by linearity, while one also has
\begin{align*}
    \Im\langle {\Id} \cdot \Psi, \chi_* \Psi \rangle
    &=
    \Im\langle {\Id} \cdot \Psi^+, \chi_* \Psi^+ \rangle
    +
    \Im\langle {\Id} \cdot \Psi^-, \chi_* \Psi^- \rangle,
    \\[0.2cm]
    \mathfrak{T}^{\mathrm{D}}[\Psi] &= \mathfrak{T}^{\mathrm{D}}[\Psi^+] + \mathfrak{T}^{\mathrm{D}}[\Psi^-],
\end{align*}
since the chiral subbundles $\Sigma_\pm M$ are totally isotropic with respect to the inner product induced by \eqref{eq-spinor-inner-product}, and the subbundles $\mathscr{S}_\pm$ are orthogonal (by definition).
For each of these operators, the terms on the right-hand side can be calculated by applying the formulae from \S \ref{sec-dirac-operator-formulas} to each term separately. More precisely:
\begin{itemize}[itemsep=0.1cm]
    \item the $\Psi^+$-terms are calculated by setting the components $\xi_-=\eta_-=0$, the parameters $(d,h)=(d_+,h_+)$ so that the basis vectors are $v_\ell=v_\ell^+$, and the invariant index is $\ell=\ell_+$,
    \item the $\Psi^-$-terms are calculated by setting the components $\xi_+=\eta_+=0$, the parameters $(d,h)=(d_-,h_-)$ so that the basis vectors are $v_\ell=v_\ell^-$, and the invariant index $\ell = \ell_-$.
\end{itemize}
See also Appendix \ref{appendix-operator-formulas}.

\subsection{Relation to known literature}
\label{sec-relation-to-literature}

With these considerations, one can now recover the static spherically symmetric ansätze from the literature.
Let us name a few.

\begin{enumerate}
    \item Setting the bundle parameters $n = m = 0$ and taking $a = 0$ (and $z=0$ which is forced in this case due to $n\not=1$) for the connection \eqref{eq-sph-sym-connection} the Yang-Mills sector reduces to the classical Maxwell model and the structure group reduces to $\U(1)$. 
    The Higgs sector also trivializes substantially, while there is no spherically symmetric (in the sense of \S \ref{sec-symmetry-general}) ansatz for the Dirac sector. That being said, one can still make a "weakly" spherically symmetric ansatz in terms of spin-weighted spherical harmonics, cf.\ e.g.\ \cite{sph-sym-standard-model, kain-edm-wormholes}.
    \item Setting $n = 1$ and $m = 0$, the Yang-Mills sector reduces to the classical spherically symmetric $\SU(2)$ Yang-Mills model that is ubiquitous in literature and often referred to as the \emph{Witten ansatz}. 
    Most notably this ansatz was used in the the entire $\SU(2)$ Einstein-Yang-Mills programme \cite{bart-mckin, bizon-colored-bh, smol-wass-1, breit-forg-mais} and many other papers. 
    In this case the Higgs field admits no consistent ansatz when the representation parameter $p = 1$ (i.e.\ the fundamental representation), but one can make an ansatz e.g.\ when $p = 2$ (i.e.\ the adjoint representation), which is also already well-known, see e.g.\ \cite{kunzle}.
    The Dirac sector admits a consistent ansatz most notably when the twisting representation is fundamental (i.e.\ $d = 1$), and this leads to the ansatz studied (among others) by \cite{finster-smoller-yau-edym-solutions, finster-smoller-yau-edym-bh-nonexistence}, though note that they use the Dirac representation rather than the Weyl representation for spinors which needs to be taken into account if one wishes to compare the formulas.
    \item Setting $n = m$ and $p = q = k = 1$, the coupled Yang-Mills and Higgs sectors are precisely the ones studied in \cite{volkov-electroweak, volkov-electroweak-blackholes}. 
    In particular, the ansatz for the Higgs field is in this case sometimes referred to as the \emph{Cho-Maison monopole}. 
    In this case one can also get a consistent ansatz for the Dirac sector but the author is not aware of any references where this has been done so we leave this discussion for the future sections.
\end{enumerate}

\section{The spherically symmetric model}
\label{sec-sphsym-model}

\subsection{General assumptions}
\label{sec-sph-sym-general-parameters}
Before writing down the Euler-Lagrange equations, let us briefly summarize the ansatz.
We denote the principal $\SU(2)\times\U(1)$-bundle with structure parameters $(n,m)$ by $P_{(n,m)}$, cf.\ \S \ref{sec-sph-sym-principal-bundles}.
We assume for simplicity that $n>0$ although the considerations from this section can be adapted also to the less interesting case $n=0$ (in fact, the only really interesting case is $n=m=1$).
We have natural section $\tilde{\sigma} = ((t,s), [\sigma(\theta,\varphi), e])$ of $P_{(n,m)}$, where $\sigma$ is the Hopf section (\ref{eq-hopf-section}).
The most general spherically symmetric connection on $P_{(n,m)}$ is given by
\begin{align*}
    \tilde{\sigma}^*\omega =&\, \mathbf{e}^t \otimes e^{-\alpha(s)}(a(s)\tau_3 + b(s)\tau_4)
    + \mathbf{e}^\varphi \otimes \frac{\cot\theta}{r(s)}\,(n\tau_3 + m\tau_4)
    \\
    &+ \frac{\Re(z)(s)}{r(s)} \,(\mathbf{e}^\theta \otimes \tau_1 + \mathbf{e}^\varphi \otimes \tau_2) + \frac{\Im(z)(s)}{r(s)} \,(\mathbf{e}^\theta \otimes \tau_2 - \mathbf{e}^\varphi \otimes \tau_1),
\end{align*}
where $a,b$ are real-valued functions and $z$ is a complex-valued function such that $z\equiv0$ if $n\not=1$, cf.\ \S \ref{sec-yang-mills-ansatz}.
For the Higgs field and the Dirac field, we consider representation parameters $(p,q)$ and $(d_\pm, h_\pm)$ as well as invariant indices $k, \ell_+, \ell_-$ respectively, such that:
\begin{itemize}[itemsep=0.2cm]
    \item $p+d_+=d_-$ and $q+h_+=h_-$,
    \item $k := \frac12\left(p+\frac{qm}{n}\right) \in \{0,\ldots,p\}$,
    \item $\ell_+ := \frac12\left(d_++\frac{h_+m+1}{n}\right) \in \{0,\ldots,d_+\}$.
\end{itemize}
We then set
\begin{equation*}
    \ell_- := \frac12\left(d_-+\frac{h_-m+1}{n}\right) = k+\ell_+ \in \{0,\ldots,d_-\}.
\end{equation*}
With these parameters define the Higgs vector space $W$ and the $\pm$-fermionic vector spaces $V_\pm$ in the same way as in Definition \ref{def-parameters}, and we also define the associated vector budles
\begin{equation*}
    \mathscr{H} = P_{(n,m)} \times_{\rho_{(p,q)}} W,
    \qquad
    \mathscr{S}_{\pm} = P_{(n,m)} \times_{\rho_{(d_\pm,h_\pm)}} V_\pm.
\end{equation*}
Then we can get a consistent ansatz for the Higgs field $\Phi \in \Gamma(\mathscr{H})$ by setting
\begin{equation*}
    \Phi_{(t,s,k_0H)} = [((t,s), [k_0, e]), \; \phi(s)\, w_k], \qquad k_0 \in K.
\end{equation*}
for a complex-valued function $\phi$, cf.\ \ref{sec-higgs-ansatz}.
Similarly, for the Dirac field we set $\Psi = \Psi^+ + \Psi^-$, where the twisted chiral parts $\Psi^\pm \in \Gamma(\Sigma_\pm M \otimes \mathscr{S}_\pm)$ are defined by
\begin{align*}
    \Psi^+_{(t,s,k_0H)} &= \left[((t,s), [k_0, e]), \; \xi_+(s)\, \u_+ \otimes v^+_{\ell_+} + \eta_+(s)\, \d_+ \otimes v^+_{\ell_+-1}\right],\\[0.1cm]
    \Psi^-_{(t,s,k_0H)} &= \left[((t,s), [k_0, e]), \; \xi_-(s)\, \u_- \otimes v^-_{\ell_-} + \eta_-(s)\, \d_- \otimes v^-_{\ell_--1}\right],
\end{align*}
for complex-valued functions $\xi_\pm, \eta_\pm$,
where it is understood that $\eta_+=0$ (resp.\ $\eta_-=0$) if $\ell=0$ (resp.\ $j=0$), 
and also $\eta_\pm = 0$ if $n > 1$, cf.\ \S \ref{sec-dirac-ansatz}.
We note finally that the Yukawa map associated to the chosen parameters is then given by
\begin{align}
    \Y_\Phi\Psi =& -im_\Y \Y_{k,\ell_+}^{\ell_-} \left( \phi\xi_+ \u_+ \otimes v_{\ell_-}^- + \conj{\phi}\xi_- \u_- \otimes v_{\ell_+}^+ \right)
    \nonumber
    \\
    &-im_\Y \Y_{k,\ell_+-1}^{\ell_--1} \left( \phi\eta_+\d_+ \otimes v_{\ell_--1}^- + \conj{\phi}\eta_- \d_- \otimes v_{\ell_+-1}^+ \right)
    \label{eq-yukawa-map-general}
\end{align}
where one interprets $\Y_{k,\ell_+-1}^{\ell_--1}=0$ if $\ell_+=0$ or $\ell_-=0$ (resp.\ $\Y_{k,\ell_+}^{\ell_-}=0$ if $\ell_+=d_++1$ or $\ell_-=d_-+1$), cf.\ Definition \ref{def-yukawa}.

\subsection{Equations}
\label{sec-total-eqs}

Having fixed all the parameters and defined all the required fields, we are now ready to write down the reduced Euler-Lagrange equations (\ref{eq-yang-mills}--\ref{eq-einstein}). 
Here we only write out the resulting equations themselves, but we provide complete formulas for the operators in Appendix \ref{appendix-operator-formulas} from which it can more clearly be seen where each equation comes from.


The Dirac equations are
\footnotesize
\begin{subequations}
\begin{empheq}[left=\empheqlbrace]{align}
    0=&\, 
    - e^{-\beta} \left[\xi_+' + \left(\frac12\alpha' + \frac{r'}{r}\right)\xi_+ \right]
    - ie^{-\alpha} \zeta^{(d_+, h_+)}_{\ell_+} \xi_+ 
    + i\sqrt{\ell_+(d_+ - \ell_+ +1)}\,\frac{z\eta_+}{r} - im_\Y \Y_{k,\ell_+}^{\ell_-} \conj{\phi} \xi_-
    \label{eq-dirac-xi+}
    \\[0.2cm]
    0=& \,
    e^{-\beta} \left[\xi_-' + \left(\frac12\alpha' + \frac{r'}{r}\right)\xi_- \right] - ie^{-\alpha} \zeta^{(d_-, h_-)}_{\ell_-} \xi_-
    - i\sqrt{\ell_-(d_- - \ell_- + 1)}\,\frac{z\eta_-}{r} - im_\Y \Y_{k,\ell_+}^{\ell_-} \phi\xi_+ 
    \label{eq-dirac-xi-}
    \\[0.2cm]
    0=& \,
    e^{-\beta} \left[\eta_+' + \left(\frac12 \alpha' + \frac{r'}{r}\right)\eta_+\right]  -ie^{-\alpha}\zeta^{(d_+, h_+)}_{\ell_+-1} \eta_+ 
    + i\sqrt{\ell_+(d_+ - \ell_+ + 1)}\,\frac{\conj{z}\xi_+}{r} - im_\Y \Y_{k,\ell_+-1}^{\ell_- - 1} \conj{\phi}\eta_-
    \label{eq-dirac-eta+}
    \\[0.2cm]
    0=& \,
    -e^{-\beta} \left[\eta_-' + \left(\frac12 \alpha' + \frac{r'}{r}\right)\eta_-\right]  -ie^{-\alpha} \zeta^{(d_-,h_-)}_{\ell_- - 1}\eta_- 
    - i\sqrt{\ell_-(d_- - \ell_- + 1)}\,\frac{\conj{z}\xi_-}{r} - im_\Y \Y_{k,\ell_+ - 1}^{\ell_- - 1} \phi\eta_+.
    \label{eq-dirac-eta-}
\end{empheq}
\end{subequations}
\normalsize
The Higgs equation is
\footnotesize
\begin{subequations}
\begin{empheq}[left=\empheqlbrace]{align}
0 = &\,  e^{-2\beta} \bigg[ \phi'' + \left(\alpha' - \beta' + \frac{2r'}{r} \right)\phi' \bigg]
    + \left[e^{-2\alpha}  (\zeta^{(p,q)}_k)^2  
    - \frac{\kappa|z|^2}{2r^2} + \lambda^2 - |\phi|^2 \right] \phi
    \nonumber
    \\
    &- m_\Y \left(\Y_{k,\ell_+}^{\ell_-}\, \conj{\xi}_+ \xi_- + \Y_{k,\ell_+ - 1}^{\ell_- - 1}\, \conj{\eta}_+\eta_-\right),
    \label{eq-higgs-phi}
\end{empheq}
\end{subequations}
\normalsize
where we recall that $\kappa$ is given by \eqref{eq-kappa}.
For the Yang-Mills sector, we first isolate the first-order equations
\footnotesize
\begin{subequations}
\begin{empheq}[left=\empheqlbrace]{align}
    0 &= \frac{4}{r^2} \Im(\conj{z}z') - (p-2k) \,\Im(\conj{\phi}\phi') 
    \nonumber
    \\
    &\quad - \frac{1}{2} e^\beta \left[ (d_+ - 2\ell_+)|\xi_+|^2 - (d_+ - 2\ell_+ + 2) |\eta_+|^2 - (d_- - 2\ell_-)|\xi_-|^2 + (d_- - 2\ell_- + 2)|\eta_-|^2 \right]
    \label{eq-ym-constraint-1}
    \\[0.2cm]
    0 &= q \Im(\conj{\phi}\phi') + \frac{1}{2} e^\beta \left[ h_+(|\xi_+|^2 - |\eta_+|^2) + h_- (- |\xi_-|^2 + |\eta_-|^2)\right],
    \label{eq-ym-constraint-2}
\end{empheq}
\end{subequations}
\normalsize
which can be viewed as constraint equations, as we will soon see.
The remaning Yang-Mills equations are
\footnotesize
\begin{subequations}
\begin{empheq}[left=\empheqlbrace]{align}
    0 =& \, - \frac{1}{r^2}e^{-\beta} (r^2 e^{-\alpha-\beta}a')' + \frac{2}{r^2} e^{-\alpha} a|z|^2 - \frac12(p-2k) e^{-\alpha} \,\zeta^{(p,q)}_k |\phi|^2 \nonumber
    \\
    & + \frac{1}{4}  \left[ (d_+ - 2\ell_+)|\xi_+|^2 + (d_+ - 2\ell_+ + 2) |\eta_+|^2 + (d_- - 2\ell_-)|\xi_-|^2 + (d_- - 2\ell_- + 2)|\eta_-|^2 \right]
    \label{eq-ym-a}
    \\[0.2cm]
    0 =& \,
    - \frac{1}{r^2}e^{-\beta} (r^2 e^{-\alpha-\beta}b')' -\frac{q}{2} e^{-\alpha} \,\zeta^{(p,q)}_k |\phi|^2
    +\frac{1}{4} \left[ h_+(|\xi_+|^2 + |\eta_+|^2) + h_-(|\xi_-|^2 + |\eta_-|^2) \right]
    \label{eq-ym-b}
    \\[0.2cm]
    0 =& \,
    e^{-\alpha-\beta}(e^{\alpha-\beta}z')' + e^{-2\alpha}a^2z + \frac{1}{r^2}z(n-|z|^2) 
    - \frac{\kappa}{4} z|\phi|^2 
    \nonumber
    \\
    &+ \frac{r}{2}  \left[ \sqrt{\ell_+(d_+ - \ell_+ + 1)}\, \xi_+\conj{\eta}_+ - \sqrt{\ell_-(d_- - \ell_- + 1)}\, \xi_-\conj{\eta}_- \right]
    \label{eq-ym-z}
\end{empheq}
\end{subequations}
\normalsize

Finally, for the Einstein sector we also isolate the first-order equations
\footnotesize
\begin{subequations}
\begin{empheq}[left=\empheqlbrace]{align}
    0 =& \, \frac{2}{r^2} a\Im(\conj{z}z')  + \zeta_k^{(p,q)} \Im(\conj{\phi}\phi') 
    \nonumber
    \\
    &+ \frac12 e^{\beta}\bigg( \zeta^{(d_+,h_+)}_{\ell_+}|\xi_+|^2 - \zeta^{(d_-,h_-)}_{\ell_-}|\xi_-|^2 
    - \zeta^{(d_+,h_+)}_{\ell_+-1} |\eta_+|^2  + \zeta^{(d_-,h_-)}_{\ell_--1}|\eta_-|^2
    \bigg) 
    \label{eq-einstein-constraint-1}
    \\[0.2cm]
    0 =& \,
    1 - r^2e^{-2\beta}\left(2\alpha' + \frac{r'}{r}\right)\frac{r'}{r}
    \nonumber
    \\
    &- \frac12 r^2e^{-2\alpha-2\beta} (a'{}^2 + b'{}^2)
    +  e^{-2\alpha}a^2|z|^2 + e^{-2\beta}|z'|^2 - \frac{(n-|z|^2)^2+m^2}{2r^2}
    \nonumber
    \\
    & + \frac12 r^2 e^{-2\alpha}|\zeta_k^{(p,q)}\phi|^2 + \frac12 r^2 e^{-2\beta}|\phi'|^2 - \frac14 \kappa|z|^2|\phi|^2 - \frac12 r^2 U(\phi)
    \nonumber
    \\[0.1cm]
    &- \frac12 r^2 e^{-\alpha} \left[ \zeta^{(d_+,h_+)}_{\ell_+} |\xi_+|^2 +  \zeta^{(d_-,h_-)}_{\ell_-}|\xi_-|^2
    + \zeta^{(d_+,h_+)}_{\ell_+ - 1} |\eta_+|^2 +  \zeta^{(d_-,h_-)}_{\ell_- - 1}|\eta_-|^2 \right]
    \nonumber
    \\[0.1cm]
    &+ \sqrt{\ell_+(d_+ - \ell_+ + 1)}\, r \Re(z\conj{\xi}_+\eta_+) 
    - \sqrt{\ell_-(d_- - \ell_- + 1)}\, r \Re(z\conj{\xi}_-\eta_-)
    \nonumber
    \\[0.1cm]
    &- m_\Y \Y_{k,\ell_+}^{\ell_-} r^2 \Re(\phi\conj{\xi}_-\xi_+)
    - m_\Y \Y_{k, \ell_+ - 1}^{\ell_- - 1} r^2 \Re(\phi\conj{\eta}_-\eta_+),
    \label{eq-einstein-constraint-2}
\end{empheq}
\end{subequations}
\normalsize
which can again be viewed as constraints.
In fact, we note already here that \eqref{eq-einstein-constraint-1} is just a linear combination of the Yang-Mills constraints (\ref{eq-ym-constraint-1}) and (\ref{eq-ym-constraint-2}), so that omit it henceforth.
The remaining Einstein equations are
\footnotesize
\begin{subequations}
\begin{empheq}[left=\empheqlbrace]{align}
    0 =& \,e^{-2\beta}\left[ \frac{r''}{r} + \left(\alpha' - \beta' + \frac{r'}{r}\right)\frac{r'}{r} \right] - \frac{1}{r^2}
    \nonumber
    \\[0.1cm]
    &\, + \frac12 e^{-2\alpha-2\beta}(a'{}^2 + b'{}^2) + \frac{(n-|z|^2)^2 + m^2}{2r^4} + \frac{\kappa|z|^2 |\phi|^2}{4r^2} + \frac12 U(\phi)
    \nonumber
    \\[0.1cm]
    &\, - \frac{1}{2r}\sqrt{\ell_+(d_+ - \ell_+ + 1)} \, \Re(z\conj{\xi}_+\eta_+) + \frac{1}{2r}\sqrt{\ell_-(d_- - \ell_- + 1)} \, \Re(z\conj{\xi}_-\eta_-)
    \nonumber
    \\
    &\, + \frac12 m_\Y \Y_{k,\ell_+}^{\ell_-} \Re(\phi\conj{\xi}_-\xi_+) + \frac12 m_\Y \Y_{k,\ell_+-1}^{\ell_- - 1} \Re(\phi\conj{\eta}_-\eta_+),
    \label{eq-einstein-evo-ddr}
    \\[0.2cm]
    0 =&\, -e^{-2\beta}\left[\alpha'' + \left( \alpha' - \beta' + \frac{2r'}{r} \right)\alpha'\right]
    \nonumber
    \\
    &+ \frac12 e^{-2\alpha-2\beta} (a'{}^2 + b'{}^2)
    + \frac{1}{r^2}\left (e^{-2\alpha}a^2|z|^2 + e^{-2\beta}|z'|^2 
    + \frac{(n-|z|^2)^2+m^2}{2r^2} \right)
    + e^{-2\alpha}|\zeta_k^{(p,q)}\phi|^2
    \nonumber
    \\
    & - \frac12 U(\phi) -\frac12 e^{-\alpha}\left( \zeta^{(d_+,h_+)}_{\ell_+} |\xi_+|^2 +  \zeta^{(d_-,h_-)}_{\ell_-}|\xi_-|^2
    + \zeta^{(d_+,h_+)}_{\ell_+ - 1} |\eta_+|^2 +  \zeta^{(d_-,h_-)}_{\ell_- - 1}|\eta_-|^2\right)
    \nonumber
    \\
    & - \frac12 m_\Y \Y_{k,\ell_+}^{\ell_-} \Re(\phi\conj{\xi}_-\xi_+) - \frac12 m_\Y \Y_{k,\ell_+-1}^{\ell_- - 1} \Re(\phi\conj{\eta}_-\eta_+)
    \label{eq-einstein-evo-ddt}
\end{empheq}
\end{subequations}
\normalsize
where we have implicitly used the constraint equation to simplify, and isolated the second-order derivatives.

\subsection{Constraints and invariances}
\label{sec-constraints-invariances}

At first glance, it may seem that the total equation system from the previous section is overdetermined.
However, equations (\ref{eq-ym-constraint-1}, \ref{eq-ym-constraint-2}, \ref{eq-einstein-constraint-1}, \ref{eq-einstein-constraint-2}) are constraints, as already hinted, in the following sense.

\begin{proposition}
\label{prop-constraint-propagation}
    \leavevmode
    \begin{enumerate}
        \item Assume the Yang-Mills equations {\normalfont(\ref{eq-ym-a}--\ref{eq-ym-z})}, Higgs equation \eqref{eq-higgs-phi}, and Dirac equations {\normalfont(\ref{eq-dirac-xi+}--\ref{eq-dirac-eta-})} are satisfied. 
        If the Yang-Mills constraint equations {\normalfont (\ref{eq-ym-constraint-1}, \ref{eq-ym-constraint-2})} are satisfied at some point $s = s_0$, then they are satisfied for all $s$ for which the solution exists.
        \item With the same assumptions as in {\normalfont (i) above}, assume additionally the Einstein equations {\normalfont (\ref{eq-einstein-evo-ddr}, \ref{eq-einstein-evo-ddt})} are satisfied.
        If the Einstein constraint equations {\normalfont (\ref{eq-einstein-constraint-1}, \ref{eq-einstein-constraint-2})} are satisfied at some point $s = s_0$, then they are satisfied for all $s$ for which the solution exists.
    \end{enumerate}
\end{proposition}

\begin{remark}
\label{rem-einstein-constraints}
    Concerning the Einstein equations, we have already noted that (\ref{eq-einstein-constraint-1}) is just a combination of the Yang-Mills constraints (\ref{eq-ym-constraint-1}, \ref{eq-ym-constraint-2}).
    Of the remaining three Einstein equations (\ref{eq-einstein-constraint-2}, \ref{eq-einstein-evo-ddr}, \ref{eq-einstein-evo-ddt}) one can actually choose to solve any two, and show that the third one will be automatically implied, which is a slightly more general statement than the one given in Proposition \ref{prop-constraint-propagation} (ii).
    However it turns out that is generally beneficial for the analysis to solve the second-order equations (\ref{eq-einstein-evo-ddr}, \ref{eq-einstein-evo-ddt}), and keep the first-order equation (\ref{eq-einstein-constraint-2}) as a constraint, and therefore do not bother to prove the more general statement here.
\end{remark}

\begin{proof}
    Let us first treat the Yang-Mills constraints (\ref{eq-ym-constraint-1}, \ref{eq-ym-constraint-2}).
    To this end, set
    \begin{align*}
        F &= \frac{4}{r^2} \Im(\conj{z}z') - (p-2k) \,\Im(\conj{\phi}\phi') 
        \\
        &\quad - \frac{1}{2} e^\beta \left[ (d_+ - 2\ell_+)|\xi_+|^2 - (d_+ - 2\ell_+ + 2) |\eta_+|^2 - (d_- - 2\ell_-)|\xi_-|^2 + (d_- - 2\ell_- + 2)|\eta_-|^2 \right]
        \\[0.2cm]
        G &= q \Im(\conj{\phi}\phi') + \frac{1}{2} e^\beta \left[ h_+(|\xi_+|^2 - |\eta_+|^2) + h_- (- |\xi_-|^2 + |\eta_-|^2)\right]
    \end{align*}
    A calculation shows
    \begin{align*}
        \left[\frac{1}{r^2} \Im(\conj{z}z')\right]' =
        &\, -\left(\alpha' - \beta' + \frac{2r'}{r}\right)\frac{1}{r^2}\Im(\conj{z}z')
        \\
        &\, + \frac{1}{2r} e^{2\beta} \left[ \sqrt{\ell_+(d_+ - \ell_+ + 1)}\Im(z\conj{\xi}_+\eta_+) - \sqrt{\ell_-(d_- - \ell_- + 1)}\Im(z\conj{\xi}_-\eta_-) \right]
        \\[0.2cm]
        \left[ \Im(\conj{\phi}\phi') \right]'
        =
        &\, -\left(\alpha' - \beta' + \frac{2r'}{r}\right)\Im(\conj{\phi}\phi')
        - m_\Y e^{2\beta} \left[ \Y_{k,\ell_+}^{\ell_-} \Im(\phi \conj{\xi}_-\xi_+) + \Y_{k,\ell_+-1}^{\ell_--1} \Im(\phi \conj{\eta}_-\eta_+) \right]
        \\[0.2cm]
        \left[\frac{e^\beta}{2}|\xi_\pm|^2\right]'
        =
        &\, -\left(\alpha' - \beta' + \frac{2r'}{r}\right) \frac{e^\beta}{2}|\xi_\pm|^2
        \\
        &\, - \sqrt{\ell_\pm(d_\pm - \ell_\pm + 1)}\, \frac{e^{2\beta}}{r} \Im(z\conj{\xi}_\pm\eta_\pm)
        - m_\Y \Y_{k,\ell_+}^{\ell_-} e^{\beta} \Im(\phi\conj{\xi}_-\xi_+)
        \\[0.2cm]
        \left[\frac{e^\beta}{2}|\eta_\pm|^2\right]'
        =
        &\, -\left(\alpha' - \beta' + \frac{2r'}{r}\right) \frac{e^\beta}{2}|\eta_\pm|^2
        \\
        &\, - \sqrt{\ell_\pm(d_\pm - \ell_\pm + 1)}\, \frac{e^{2\beta}}{r} \Im(z\conj{\xi}_\pm\eta_\pm)
        + m_\Y \Y_{k,\ell_+-1}^{\ell_--1} e^{\beta} \Im(\phi\conj{\eta}_-\eta_+).
    \end{align*}
    Then due to the relations between the parameters $(p,q)$, $(d_+,h_+)$, $(d_-,h_-)$ and the invariant indices $k$, $\ell_+$, $\ell_-$, we get
    \begin{align*}
        F' =&\, -\left(\alpha' - \beta' + \frac{2r'}{r}\right) F
        \\
        &\,  + [2 + (d_+-2\ell_+) - (d_+ - 2\ell_+ + 2)]\sqrt{\ell_+(d_+ - \ell_+ + 1)}\,\frac{e^{2\beta}}{r}\Im(z\conj{\xi}_+\eta_+)
        \\
        &\, - [2 + (d_--2\ell_-) - (d_- -2\ell_- + 2)]\sqrt{\ell_-(d_- - \ell_- + 1)}\,\frac{e^{2\beta}}{r}\Im(z\conj{\xi}_-\eta_-)
        \\
        &\, + [(p - 2k)  + (d_+ - 2\ell_+) - (d_- - 2\ell_-)] \, m_\Y \Y_{k,\ell_+}^{\ell_-} e^{2\beta}\Im(\phi \conj{\xi}_-\xi_+)
        \\
        &\, + [(p - 2k)  + (d_+ - 2\ell_+ + 2) - (d_- - 2\ell_- + 2)] \, m_\Y \Y_{k,\ell_+}^{\ell_-} e^{2\beta}\Im(\phi \conj{\eta}_-\eta_+)
        \\[0.2cm]
        =&\, -\left(\alpha' - \beta' + \frac{2r'}{r}\right) F,
    \end{align*}
    and
    \begin{align*}
        G' =&\, -\left(\alpha' - \beta' + \frac{2r'}{r}\right) G
        + (-q - h_+ + h_-)\, m_\Y  e^{2\beta} \left[ \Y_{k,\ell_+}^{\ell_-} \Im(\phi \conj{\xi}_-\xi_+) + \Y_{k,\ell_+-1}^{\ell_--1} \Im(\phi \conj{\eta}_-\eta_+) \right]
        \\[0.2cm]
        =&\, -\left(\alpha' - \beta' + \frac{2r'}{r}\right) G.
    \end{align*}
    This shows that $r^2 e^{\alpha-\beta} F$ and $r^2 e^{\alpha-\beta} G$ are constants of motion and it follows that (\ref{eq-ym-constraint-1}, \ref{eq-ym-constraint-2}) propagate, if satisfied initially, proving part (i).

    Next, for the Einstein constraints (\ref{eq-einstein-constraint-1}, \ref{eq-einstein-constraint-2}), we note that (\ref{eq-einstein-constraint-1}) is actually a linear combination of the Yang-Mills constraints (\ref{eq-ym-constraint-1}, \ref{eq-ym-constraint-2}), schematically
    \begin{equation*}
        {\normalfont (\ref{eq-einstein-constraint-1})} = \frac12 a \cdot  {\normalfont (\ref{eq-ym-constraint-1})}  - \frac12 b \cdot {\normalfont (\ref{eq-ym-constraint-2})}.
    \end{equation*}
    Thus this constraint propagates by part (i).
    Finally, letting $H$ denote the RHS of (\ref{eq-einstein-constraint-2}), a rather long computation (cf.\ Appendix \ref{appendix-einstein-constraint}) shows that
    \begin{equation}\label{eq-einstein-constraint-propagation-equation}
        H' = -2\left( \alpha' + \frac{r'}{r} \right)H.
    \end{equation}
    Hence $r^2e^{2\alpha} H$ is a constant of motion and it follows that (\ref{eq-einstein-constraint-2}) propagates.
\end{proof}

In particular, Proposition \ref{prop-constraint-propagation} shows that the problem of finding solutions to the total equation system from \S \ref{sec-total-eqs} is well-posed as an initial value problem for the ODEs (\ref{eq-ym-a}--\ref{eq-ym-z}, \ref{eq-higgs-phi}, \ref{eq-dirac-xi+}--\ref{eq-dirac-eta-}, \ref{eq-einstein-evo-ddr}, \ref{eq-einstein-evo-ddt}), with initial data constrained by (\ref{eq-ym-constraint-1}, \ref{eq-ym-constraint-2}, \ref{eq-einstein-constraint-1}, \ref{eq-einstein-constraint-2}).

There is another constant of motion, namely the following.

\begin{proposition}
    If the Dirac equations {\normalfont (\ref{eq-dirac-xi+}--\ref{eq-dirac-eta-})} hold, then
    \begin{equation*}
        r^2e^{\alpha}(|\xi_+|^2 - |\eta_+|^2 - |\xi_-|^2 + |\eta_-|^2) \equiv \mathrm{const}.
    \end{equation*}
\end{proposition}

\begin{proof}
    This follows directly by differenting and using the Dirac equations (\ref{eq-dirac-xi+}--\ref{eq-dirac-eta-}).
\end{proof}

\begin{remark}
    The author is not entirely certain what this constant of motion corresponds to geometrically. 
    One gets a similar constant of motion also for the Riemannian Dirac-Yang-Mills equations in spherical symmetry, cf.\ \cite{dym-adam-marko}.
\end{remark}

We also note that the equations have a $\U(1)$-invariance
\begin{gather}
\label{eq-u(1)-invariance-tau3}
    z \mapsto e^{i\mu} z,
    \qquad
    \xi_+ \mapsto e^{\frac{i\mu}{2}} \xi_+,
    \qquad
    \xi_- \mapsto e^{\frac{i\mu}{2}} \xi_-,
    \qquad
    \eta_+ \mapsto e^{-\frac{i\mu}{2}} \eta_+,
    \qquad
    \eta_- \mapsto e^{-\frac{i\mu}{2}} \eta_-,
\end{gather}
for any fixed constant $\mu \in \R$,
as well as the additional $\U(1) \times \U(1)$-invariance
\begin{gather}
\label{eq-u(1)-invariance-tau4}
    \phi \mapsto e^{i\nu} \phi
    \qquad
    \xi_+ \mapsto e^{i\tau} \xi_+,
    \qquad
    \xi_- \mapsto e^{i(\nu+\tau)} \xi_-,
    \qquad
    \eta_+ \mapsto e^{i\tau} \eta_+,
    \qquad
    \eta_- \mapsto e^{i(\nu+\tau)} \eta_-,
\end{gather}
for any fixed constants $\nu,\tau \in \R$. 
They are a consequence of the gauge-invariance of the equations and the residual gauge freedom.

\subsection{Coordinate choices}
\label{sec-coordinate-choices}

In the previous sections we have written the system down for the most general choice \eqref{eq-static-sph-sym-metric-general} of metric $g$, depending on the three functions $\alpha$, $\beta$, $r$ with $r > 0$. 
However, only two of these coefficients are truly independent, in the sense that one can without loss of generality choose $s$ in a way that the third coefficient is determined by some other desirable condition.
There are three natural choices for $s$, namely:
\begin{enumerate}
    \item the \emph{radial} coordinate system such that $s = r$, i.e. the metric is
    \begin{equation*}
        g = -e^{2\alpha(r)} \diff t^2 + e^{2\beta(r)} \diff r^2 + r^2 g_{\mathbb{S}^2},
    \end{equation*}
    and primes correspond to differentiation by $r$ so that $r' = 1$ and $r'' = 0$.
    This is the most commonly used and natural coordinate system in literature when studying spherical symmetry, but it comes at the disadvantage of breaking down at points where $r$ is no longer monotone (i.e.\ at marginally trapped spheres), as well as making the total system non-autonomous. In this case it is also common (in particular when studying black hole solutions) to replace $\alpha$ and $\beta$ by \emph{Schwarzschild-like} variables $M$ and $\delta$, which we comment on in Appendix \ref{sec-schwarzschild-coordinates}.
    \item the \emph{arc-length} coordinate system such that the independent variable $s$ is fixed by demanding that $e^\beta = 1$, so that the metric is
    \begin{equation*}
        g = -e^{2\alpha(s)} \diff t^2 + \diff s^2 + r(s)^2 g_{\mathbb{S}^2}
    \end{equation*}
    and $\beta' = 0$. In this case, the constant $t$-slices have the form of a warped product $I\times_{r}\mathbb{S}^2$ with respect to the radius function $r$. In particular, for fixed $t\in\R$ and $kH \in \mathbb{S}^2$, the curves $s \mapsto (t, s, kH)$ for fixed $t$ and $kH$ are parametrized by arc-length. 
    This system no longer suffers the same singularity problems at stationary points of $r$ and is autonomous.
    \item the \emph{log-polar} coordinate system such that the independent variable $s$ is fixed by demanding that $e^\beta = r$, so that the metric is
    \begin{equation*}
        g = -e^{2\alpha(s)} \diff t^2 + r(s)^2 (\diff s^2 + g_{\mathbb{S}^2}),
    \end{equation*}
    and $\beta' = r'/r$. In this case, the constant $t$-slices are conformally related (via the radius function $r$) to the cylinder $I \times \mathbb{S}^2$. 
    This system also does not suffer from singularities at stationary points of $r$, and has the benefit of making the equations (in particular the Yang-Mills equation) depend more simply on the metric coefficients, while also being autonomous.
\end{enumerate}

\subsection{Initial conditions}

By the invariances (\ref{eq-u(1)-invariance-tau3}, \ref{eq-u(1)-invariance-tau4}) of the equations, we can assume without loss of generality that $z$ and $\phi$ are real and non-negative initially, so that
\begin{equation*}
    z_0 = |z_0| \in \R,
    \qquad
    \phi_0 = |\phi_0| \in \R.
\end{equation*}
The electric variables of the Yang-Mills potential $a, b, a', b'$ are initially unconstrained so we may take any
\begin{equation*}
    a_0, b_0, a'_0, b'_0 \in \R.
\end{equation*}
The remaining variables can then be adjusted to satisfy the constraints.
More precisely:
\begin{itemize}
    \item If $q = 0$, then necessarily $h_+ = h_-$ and \eqref{eq-ym-constraint-2} demands that
    \begin{equation*}
        |\xi_+|^2 - |\eta_+|^2 - |\xi_-|^2 + |\eta_-|^2 = 0,
    \end{equation*}
    so in particular this must be satisfied initially. Furthermore, in this case necessarily $k = p/2$ and $p$ is even. Then \eqref{eq-ym-constraint-1} reads
    \begin{align*}
        0 =  \frac{4}{r^2} \Im(\conj{z}z') 
        - \frac{1}{2} e^\beta \Big[ & (d_+ - 2\ell_+)|\xi_+|^2 - (d_+ - 2\ell_+ + 2) |\eta_+|^2\\
        &- (d_- - 2\ell_-)|\xi_-|^2 + (d_- - 2\ell_- + 2)|\eta_-|^2 \Big].
    \end{align*}
    Thus $\phi_0' \in \C$ is also unconstrained in this case, while (since we are assuming $z_0 = |z_0|$ is real), we are forced to demand
    \begin{align*}
        |z_0|\Im (z_0') = \frac{r_0^2 e^{\beta_0}}{8} \Big[ & (d_+ - 2\ell_+)|(\xi_+)_0|^2 - (d_+ - 2\ell_+ + 2) |(\eta_+)_0|^2
        \\
        &- (d_- - 2\ell_-)|(\xi_-)_0|^2 + (d_- - 2\ell_- + 2)|(\eta_-)_0|^2 \Big],
    \end{align*}
    and $\Re(z_0') \in \R$ is free.

    \item If $q \not= 0$, we can use \eqref{eq-ym-constraint-2} to simplify \eqref{eq-ym-constraint-1} and get
    \begin{align*}
        |z_0|\Im(z_0') = - \frac{r_0^2e^{\beta_0}}{8q} \Big[&
        [(p-2k)h_+ - (d_+ - 2\ell_+)q] \, |(\xi_+)_0|^2
        \\
        &- [(p-2k)h_+ - (d_+ - 2\ell_+ + 2)q]\, |(\eta_+)_0|^2
        \\
        &- [(p-2k)h_- - (d_- - 2\ell_-)q]\, |(\xi_-)_0|^2
        \\
        &+ [(p-2k)h_- - (d_- - 2\ell_- + 2)q]\, |(\eta_-)_0|^2
        \Big].
    \end{align*}
    On the other hand $\xi_\pm, \eta_\pm$ themselves are unconstrained but $\Im(\phi_0')$ (again, since we are assuming $\phi_0 = |\phi_0|$ is real) must be determined through \eqref{eq-ym-constraint-2} via
    \begin{equation*}
        |\phi_0|\Im(\phi_0') = -\frac{1}{2q} e^{\beta_0} \left[ h_+|(\xi_+)_0|^2 - h_+|(\eta_+)_0|^2 - h_- |(\xi_-)_0|^2 + h_-|(\eta_-)_0|^2\right]
    \end{equation*}
    and $\Re(\phi_0') \in \R$ is free. 

    Assuming $n \not= 0$, using the relations between parameters we could further simplify
    \begin{align*}
        \Im(\conj{z}z') 
        &= - \frac{r^2e^{\beta}}{8n} 
        \left( |\xi_+|^2 + |\eta_+|^2 - |\xi_-|^2 - |\eta_-|^2 \right).
    \end{align*}
    In this case it is sometimes convenient to also solve the constraints for $|\xi_\pm|^2$ so that
    \begin{align*}
        |\xi_+|^2 &= - \frac{8nh_-}{qr^2} e^{-\beta} \Im(\conj{z}z') + 2e^{-\beta} \Im(\conj{\phi}\phi') - \frac{h_+ + h_-}{q} |\eta_+|^2 + \frac{2h_-}{q} |\eta_-|^2
        \\
        |\xi_-|^2 &= - \frac{8nh_+}{qr^2} e^{-\beta} \Im(\conj{z}z') + 2e^{-\beta}\Im(\conj{\phi}\phi') - \frac{2h_+}{q} |\eta_+|^2 + \frac{h_+ + h_-}{q} |\eta_-|^2.
    \end{align*}
\end{itemize}
On the other hand, the initial conditions $\alpha_0$, $\beta_0$, $r_0$ for the metric coefficients will depend additionally on the chosen coordinate system (cf.\ \ref{sec-coordinate-choices}), but in each case they will be parametrized by two real numbers.

\subsection{Trivial solutions}
\label{sec-trivial-solutions}

We would like to mention two trivial Reissner-Nordstr\"om type solutions of the system.
\begin{enumerate}
\item Set the bundle parameters to $n=m$.
Let
\begin{equation*}
    \phi \equiv \lambda e^{i\arg(\phi_0)}, \qquad \xi_+ = \xi_- = \eta_+ = \eta_- = z = 0,
\end{equation*}
and
\begin{equation*}
    a(r) = b(r) = a_0 + \frac{q_0}{r},
\end{equation*}
while the metric coefficients describe the Reissner-Nordstr\"om black hole
\begin{equation*}
    e^{2\alpha(r)} = e^{-2\beta(r)} = 1 - \frac{r_0}{r} + \frac{q_0^2 + n^2}{r^2}
\end{equation*}
where we assume without loss of generality (i.e.\ by rescaling the time variable) that $\alpha+\beta = 0$ (more generally the equations imply that this is a constant).
\item Set
\begin{equation*}
    b(r) = b_0 + \frac{q_0}{r} \qquad \xi_+ = \xi_- = \eta_+ = \eta_- = a = z = \phi = 0,
\end{equation*}
and the metric describes the Reissner-Nordstr\"om-deSitter black hole
\begin{equation*}
    e^{2\alpha(r)} = e^{-2\beta(r)} = 1 - \frac{r_0}{r} + \frac{q_0^2+n^2+m^2}{2r^2} - \frac{\lambda^4r^2}{12}. 
\end{equation*}
Moreover, if $n = 1$ then there is another such solution, obtained namely by setting
\begin{equation*}
    z = e^{i\arg(z_0)}
\end{equation*}
and keeping the remaining variables as above, in which case the metric again describes the Reissner-Nordstr\"om-deSitter black hole
\begin{equation*}
    e^{2\alpha(r)} = e^{-2\beta(r)} = 1 - \frac{r_0}{r} + \frac{q_0^2+m^2}{2r^2} - \frac{\lambda^4r^2}{12}. 
\end{equation*}
\end{enumerate}

\subsection{Further reductions}
\label{sec-further-reductions}

Given the number of equations in the total coupled system, one might instinctively want to make further simplifications and assumptions that reduce the number of variables. However, most of the usual "tricks" one uses to simplify related simpler problems do not work here, for several different reasons which we briefly discuss here.

Firstly, in $\SU(2)$ Yang-Mills-Higgs theory one often assumes that the variables $z$ and $\phi$ are purely real. 
This is perfectly valid in a Yang-Mills-Higgs theory without coupling to Dirac and can be done without loss of generality, since the constraint equations (\ref{eq-ym-constraint-1}, \ref{eq-ym-constraint-2}) imply that $\Im(\conj{z}z') = \Im(\conj{\phi}\phi') = 0$ which enables a constant gauge transformation to make $z$ and $\phi$ real-valued. 
Here, however, there are additional $\xi_\pm$ and $\eta_\pm$ terms in the constraints. 
For certain choices of parameters, these terms do vanish, as is the case for the particular ansatz in \cite{finster-smoller-yau-edym-solutions}, but in general these additional terms are truly present and non-zero in general, cf.\ e.g.\ (\ref{eq-ym-constraint-1-lepton-radial}, \ref{eq-ym-constraint-2-lepton-radial}) below.

Secondly, in $\SU(2)$ (Einstein-)Yang-Mills theory one also often assumes that the potential is in a temporal gauge, so that $a=b=0$.
This cannot be assumed by simple gauge transformations since, as noted in \S \ref{sec-yang-mills-ansatz}, this would require a $t$-dependent transformation which is not compatible with the staticity assumption.
Nevertheless, it has been shown in \cite{bizon-popp} that for the classical $\SU(2)$ Einstein-Yang-Mills system, one can still assume $a=0$ without losing generality since any asymptotically flat solution with $a \not= 0$ necessary has $z \equiv 0$, i.e.\ it must be abelian.
Here again, however, such arguments do not apply due to the additional twisted spinor terms in \eqref{eq-ym-z}, which break certain monotonicity properties that the proof in \cite{bizon-popp} is based on.

Finally, we would like to mention that for the Dirac sector one also often likes to assume that the associated variables are related via conjugations and/or are real-valued, but this cannot be assumed in general for the model considered here, most notably because of the imaginary unit in the Dirac equations (\ref{eq-dirac-xi+}--\ref{eq-dirac-eta-}), which also relates to our first point above.
Part of the reason for the lack of the well-known reductions as described above is that the model considered here is truly \emph{chiral} in the sense that the spinor fields of positive/negative chirality are twisted by vector bundles associated to non-isomorphic representations (as is the case in the Standard Model).

Therefore, the model considered here already seems to be fully reduced in the sense that any further reductions lose generality (although the considerations above of course do not prove this definitively).

\section{Lepton-like parameters}
\label{sec-lepton-like-params}

The model proposed in \S \ref{sec-sphsym-model} assumes that the representation parameters satisfying all the necessary conditions exist, without fixing any particular choice of such parameters.
We would now like to demonstrate a working choice of these parameters.
The resulting model is in many ways reminiscent of the lepton sector of the Standard Model, although it does not match it exactly due to the choice of $\U(1)$-representation parameters.
Therefore we call it \emph{lepton-like}.

\subsection{Parameters}

Firstly, we choose the structure parameters $n=m=1$.
We let
\begin{equation*}
    d_+ = 0 \quad \text{and} \quad p = d_- = 1,
\end{equation*}
so that $\SU(2)$ acts trivially on the fermionic fields of positive chirality, and via the fundamental representation on the fermionic fields of negative chirality as well as on the Higgs field.
We also let
\begin{equation*}
    q = 1, \quad h_+ = -1, \quad h_- = 0,
\end{equation*}
which describes the action of $\U(1)$ on the corresponding fields 
(herein lies the main difference between the model presented here and the actual lepton sector of the Standard Model - the action of the $\U(1)$ group is different, which then also changes the charges of the corresponding particles).
Finally, we can put
\begin{equation*}
    k = 1, \qquad \ell_+ = 0, \qquad \ell_- = 1,
\end{equation*}
so that the relations from \S \ref{sec-sph-sym-general-parameters} are all satisfied.
The fields are given (with respect to the same section $\epsilon \times \tilde{\sigma}$ as in \S \ref{sec-sphsym-model}) by
\begin{align*}
    \omega =& \, \mathbf{e}^t \otimes e^{-\alpha(s)}(a\tau_3 + b\tau_4)
    + \mathbf{e}^\varphi \otimes \frac{\cot\theta}{r}\,(\tau_3 + \tau_4)
    \\
    &+ \frac{\Re(z)}{r} \,(\mathbf{e}^\theta \otimes \tau_1 + \mathbf{e}^\varphi \otimes \tau_2) + \frac{\Im(z)}{r} \,(\mathbf{e}^\theta \otimes \tau_2 - \mathbf{e}^\varphi \otimes \tau_1),
    \\[0.2cm]
    \Phi =& \, \phi(s)\, w_1
    \\[0.2cm]
    \Psi =& \, \xi_+(s)\, \u_+ \otimes v^+_0 + \xi_-(s)\, \u_- \otimes v^-_1 + \eta_-(s)\, \d_- \otimes v^-_{0}.
\end{align*}
The Yukawa map then simplifies to
\begin{equation*}
    \Y_\Phi \Psi = -im_\Y(\phi\xi_+ \u_+ \otimes v_1^- + \conj{\phi}\xi_- \u_- \otimes v_0^+),
    \qquad
    \langle \Psi, i\Y_\Phi \Psi \rangle = 2m_\Y\Re(\phi\conj{\xi}_-\xi_+).
\end{equation*}
In particular, note that $\xi_\pm$ will give each other mass, just like the right-/left-handed electrons in the Standard Model.
On the other hand, $\eta_-$ does not appear in the Yukawa coupling (since it has no corresponding positive chirality field) and will consequently be massless, just like the neutrino in the Standard Model.

\subsection{Simplified equations}
We state here the simplified equations with respect to the chosen lepton-like parameters, and with respect to the radial coordinate $s = r$.
The matter field equations are
\footnotesize
\begin{subequations}
\begin{empheq}[left=\empheqlbrace]{align}
    0 =&\, \Im(\conj{z}z')
    + \frac{1}{8} r^2 e^\beta (|\xi_+|^2 - |\xi_-|^2 - |\eta_-|^2).
    \label{eq-ym-constraint-1-lepton-radial}
    \\[0.2cm]
    0 =&\, \Im(\conj{\phi}\phi') - \frac{1}{2} e^\beta|\xi_+|^2
    \label{eq-ym-constraint-2-lepton-radial}
    \\[0.2cm]
    0=&\, 
    \xi_+' + \left(\frac12\alpha' + \frac{1}{r}\right)\xi_+ 
    + ie^\beta\left( \frac{1}{2}e^{-\alpha}b \xi_+ 
    + m_\Y \conj{\phi} \xi_- \right)
    \label{eq-dirac-xi+-lepton}
    \\[0.2cm]
    0=& \,
    \xi_-' + \left(\frac12\alpha' + \frac{1}{r}\right)\xi_- 
    - ie^\beta\left(\frac{1}{2} e^{-\alpha} a \xi_-
    +  \frac{1}{r} z\eta_- 
    + m_\Y \phi\xi_+\right) 
    \label{eq-dirac-xi--lepton}
    \\[0.2cm]
    0=& \,
    \eta_-' + \left(\frac12 \alpha' + \frac{1}{r}\right)\eta_- 
    + ie^\beta\left(
    - \frac{1}{2}e^{-\alpha} a \eta_- 
    + \frac{1}{r}\conj{z}\xi_-\right)
    \label{eq-dirac-eta--lepton}
    \\[0.2cm]
    0 = &\,  \phi'' + \left(\alpha' - \beta' + \frac{2}{r} \right)\phi' 
    + e^{2\beta} \left[ \left(\frac14 e^{-2\alpha}(a-b)^2 
    - \frac{|z|^2}{2r^2} + \lambda^2 - |\phi|^2 \right) \phi
    - m_\Y  \conj{\xi}_+ \xi_-\right]
    \label{eq-higgs-phi-lepton}
    \\[0.2cm]
    0 =& \, a'' + \left(\frac{2}{r}-\alpha'-\beta'\right) a'
    + e^{2\beta}\left[ - \frac{2}{r^2} a|z|^2 
    - \frac14 (a-b) |\phi|^2
    + \frac{1}{4}e^{\alpha} (|\xi_-|^2 - |\eta_-|^2 ) \right]
    \label{eq-ym-a-lepton}
    \\[0.2cm]
    0 =& \,
    b'' + \left( \frac{2}{r} - \alpha' - \beta' \right)b'
    + e^{2\beta}\left[ \frac{1}{4} (a-b) |\phi|^2
    + \frac{1}{4} e^{\alpha}  |\xi_+|^2 \right] 
    \label{eq-ym-b-lepton}
    \\[0.2cm]
    0 =& \,
    z'' + (\alpha' - \beta')z'
    + e^{2\beta}\left[
    \frac{1}{r^2} z(1-|z|^2) 
    + z\left(e^{-2\alpha} a^2 - \frac{1}{4}|\phi|^2\right) 
    - \frac{1}{2} r \xi_-\conj{\eta}_- \right]
    \label{eq-ym-z-lepton}
\end{empheq}
\end{subequations}
\normalsize
and the Einstein equations are
\footnotesize
\begin{subequations}
\begin{empheq}[left=\empheqlbrace]{align}
    0 =&\, 
    1 - e^{-2\beta}(2r\alpha' + 1)
    \nonumber
    \\
    &- \frac12 r^2e^{-2\alpha-2\beta} (a'{}^2 + b'{}^2)
    +  e^{-2\alpha}a^2|z|^2 + e^{-2\beta}|z'|^2 - \frac{(1-|z|^2)^2+1}{2r^2}
    \nonumber
    \\
    & + \frac18 r^2 e^{-2\alpha} (a-b)^2|\phi|^2 + \frac12 r^2 e^{-2\beta}|\phi'|^2 - \frac14 |z|^2|\phi|^2 - \frac14 r^2 (\lambda^2 - |\phi|^2)^2
    \nonumber
    \\[0.1cm]
    &- \frac12 r^2 e^{-\alpha} \left[ b|\xi_+|^2 +  a|\xi_-|^2 - a|\eta_-|^2 \right]
    - r \Re(z\conj{\xi}_-\eta_-)
    - m_\Y r^2 \Re(\phi\conj{\xi}_-\xi_+)
    \label{eq-einstein-constraint-2-radial-lepton}
    \\[0.2cm]
    0 =& \,e^{-2\beta} \frac{1}{r}\left(\alpha' - \beta'\right) + \frac{e^{-2\beta}-1}{r^2}
    \nonumber
    \\[0.1cm]
    &\, + \frac12 e^{-2\alpha-2\beta}(a'{}^2 + b'{}^2) + \frac{(1-|z|^2)^2 + 1}{2r^4} + \frac{|z|^2 |\phi|^2}{4r^2} + \frac12 U(\phi)
    \nonumber
    \\[0.1cm]
    &\, + \frac{1}{2r}\Re(z\conj{\xi}_-\eta_-)
    + \frac12 m_\Y \Re(\phi\conj{\xi}_-\xi_+),
    \label{eq-einstein-evo-ddr-lepton}
    \\[0.2cm]
    0 =&\, -e^{-2\beta}\left[\alpha'' + \left( \alpha' - \beta' + \frac{2}{r} \right)\alpha'\right]
    \nonumber
    \\
    &+ \frac12 e^{-2\alpha-2\beta} (a'{}^2 + b'{}^2)
    + \frac{1}{r^2}\left (e^{-2\alpha}a^2|z|^2 + e^{-2\beta}|z'|^2 
    + \frac{(1-|z|^2)^2 + 1}{2r^2} \right)
    + \frac14 e^{-2\alpha}(a-b)^2|\phi|^2
    \nonumber
    \\
    & - \frac12 U(\phi) -\frac12 e^{-\alpha}( b |\xi_+|^2 +  a|\xi_-|^2
    - a |\eta_-|^2) - \frac12 m_\Y \Re(\phi\conj{\xi}_-\xi_+).
    \label{eq-einstein-evo-ddt-lepton}
\end{empheq}
\end{subequations}
\normalsize

\subsection{Abelian reduction}
\label{sec-abelian}

The lepton-like system can be further reduced (though not without loss of generality) by assuming
\begin{equation*}
    z = \eta_- = 0, \qquad a = b, \qquad \xi_- = \conj{\xi}_+ =: \xi.
\end{equation*}
This essentially turns off the weak coupling so that the model degenerates to a Dirac-Maxwell-Higgs model.
Then the SM-system reduces to (with respect to the radial coordinate)
\begin{align*}
    0 &= \Im(\conj{\phi} \phi') - \frac12 e^{\beta} |\xi|^2,
    \\
    0 &= \xi' + \left(\frac12\alpha' + \frac{1}{r}\right)\xi - ie^\beta \left( \frac12 e^{-\alpha}a\xi + m_\Y \phi \conj{\xi} \right),
    \\
    0 &= \phi'' + \left( \alpha' - \beta' + \frac{2}{r} \right)\phi' + e^{2\beta}\left[ (\lambda^2 - |\phi|^2)\phi - m_\Y \xi^2 \right],
    \\
    0 &= a'' + \left( \frac{2}{r} - \alpha' - \beta' \right) a' + \frac14 e^{\alpha+2\beta} |\xi|^2,
\end{align*}
and the Einstein equations are
(in terms of the Schwarzschild variables, cf.\ Appendix \ref{sec-schwarzschild-coordinates})
\begin{align*}
    M' =&\, \frac12 r^2 e^{2\delta} a'{}^2 + \frac{1}{2r^2}
        + \frac14 r^2 N |\phi'|^2
        + \frac14 r^2 U(\phi)
        + \frac14 N^{-\frac12} r^2 e^\delta a |\xi|^2
    \\[0.2cm]
    \delta' =& 
        - \frac12 r|\phi'|^2
        + \frac12 re^{\delta} N^{-\frac32} a|\xi|^2
        + \frac12 m_\Y rN^{-1}  \Re(\conj{\phi}\xi^2)
\end{align*}

Note that the trival Reissner-Nordström (RN) type solutions from \S \ref{sec-trivial-solutions} (with $z=0$) fall into this category.
Therefore, it may be beneficial to first understand this system well, in particular perturbing it around the known RN type solutions, before diving into the analysis of the general system.

\begin{appendix}

\section{Operator formulas}
\label{appendix-operator-formulas}

In \S \ref{sec-total-eqs} we presented the total reduced equation system, i.e.\ the reduction of (\ref{eq-yang-mills}--\ref{eq-einstein}) to static spherical symmetry with the ansätze derived in \S \ref{sec-static-spherical-symmetry}.
The goal of this appendix is to present the full formulas for the operators appearing in the general system (\ref{eq-yang-mills}--\ref{eq-einstein}), so that one can see which component of the fields each equation from \S \ref{sec-total-eqs} originates from.

\subsection{Dirac-Yukawa sector}

Using the formulae from \S \ref{sec-dirac-operator-formulas} and the considerations from \S \ref{sec-chiral}, we see that the Dirac-Yukawa operator then has components
\footnotesize
\begin{align*}
    \dirac_\omega&\Psi+\Y_\Phi\Psi 
    \\[0.2cm]
    =&\mathrel{\phantom{+}} \bigg\{ 
    - e^{-\beta} \left[\xi_+' + \left(\frac12\alpha' + \frac{r'}{r}\right)\xi_+ \right]
    - ie^{-\alpha} \zeta^{(d_+, h_+)}_{\ell_+} \xi_+ 
    + i\sqrt{\ell_+(d_+-\ell_++1)}\,\frac{z\eta_+}{r} - im_\Y \Y_{k,\ell_+}^{\ell_-} \conj{\phi} \xi_- \bigg\} \; \u_- \otimes v_{\ell_+}^+
    \\[0.2cm]
    &+ \bigg\{ e^{-\beta} \left[\xi_-' + \left(\frac12\alpha' + \frac{r'}{r}\right)\xi_- \right] - ie^{-\alpha} \zeta^{(d_-, h_-)}_{\ell_-} \xi_-
    - i\sqrt{\ell_-(d_--\ell_-+1)}\,\frac{z\eta_-}{r} - im_\Y \Y_{k,\ell_+}^{\ell_-} \phi\xi_+ \bigg\} \; \u_+ \otimes v_{\ell_-}^-
    \\[0.2cm]
    &+\bigg\{ e^{-\beta} \left[\eta_+' + \left(\frac12 \alpha' + \frac{r'}{r}\right)\eta_+\right]  -ie^{-\alpha}\zeta^{(d_+, h_+)}_{\ell_+-1} \eta_+ 
    + i\sqrt{\ell_+(d_+-\ell_++1)}\,\frac{\conj{z}\xi_+}{r} - im_\Y \Y_{k,\ell_+-1}^{\ell_--1} \conj{\phi}\eta_- \bigg\} \; \d_- \otimes v_{\ell_+-1}^+
    \\[0.2cm]
    &+\bigg\{-e^{-\beta} \left[\eta_-' + \left(\frac12 \alpha' + \frac{r'}{r}\right)\eta_-\right]  -ie^{-\alpha} \zeta^{(d_-,h_-)}_{\ell_--1}\eta_- 
    - i\sqrt{\ell_-(d_--\ell_-+1)}\,\frac{\conj{z}\xi_-}{r} - im_\Y \Y_{k,\ell_+-1}^{\ell_--1} \phi\eta_+ \bigg\}\; \d_+ \otimes v_{\ell_--1}^-.
\end{align*}
\normalsize

\subsection{Higgs sector}

Firstly, we note that the Yukawa current can directly be computed from (\ref{eq-yukawa-current}, \ref{eq-yukawa-map-general}) as
\begin{equation*}
    \langle \Psi, i\Y^- \Psi \rangle = m_\Y \left(\Y_{k,\ell_+}^{\ell_-}\, \conj{\xi}_+ \xi_- + \Y_{k,\ell_+-1}^{\ell_--1}\, \conj{\eta}_+\eta_-\right) w_k
\end{equation*}
(note that $k$ is fixed here, i.e.\ no summation).
The calculations from \S \ref{sec-higgs-ansatz} imply that
\begin{align*}
    \Box_\omega \Phi & - \frac12 \grad U_\Phi - \langle \Psi, i\Y^-\Psi \rangle
    \\[0.1cm]
    = \bigg\{&e^{-2\beta} \bigg[ \phi'' + \left(\alpha' - \beta' + \frac{2r'}{r} \right)\phi' \bigg]
    \\
    &+ \left[e^{-2\alpha} (\zeta^{(p,q)}_k)^2   
    - \frac{\kappa|z|^2}{2r^2} + \lambda^2 - |\phi|^2 \right] \phi
    - m_\Y \left(\Y_{k,\ell_+}^{\ell_-}\, \conj{\xi}_+ \xi_- + \Y_{k,\ell_+-1}^{\ell_--1}\, \conj{\eta}_+\eta_-\right)\bigg\} \, w_k,
\end{align*}
where we recall that $\kappa = k(p-k+1) + (k+1)(p-k)$.

\subsection{Yang-Mills sector}

From \eqref{eq-higgs-covariant-derivative}, we compute
\begin{align*}
    \Re\langle\nabla_\omega\Phi \otimes \rho_*\Phi \rangle
    =&
    \left[ e^{-\alpha} \,\zeta^{(p,q)}_k |\phi|^2 \, \mathbf{e}^t + e^{-\beta} \,\Im(\conj{\phi}\phi') \,\mathbf{e}^s \right] \otimes \frac{q}{2}\left(\frac{m}{n}\tau_3 - \tau_4\right)
    \\[0.1cm]
    &+ \frac{\kappa|\phi|^2}{4r}\left[ \Re(z)(\mathbf{e}^\theta \otimes \tau_1 + \mathbf{e}^\varphi \otimes \tau_2) + \Im(z)(\mathbf{e}^\theta\otimes\tau_2 - \mathbf{e}^\varphi\otimes\tau_1) \right].
\end{align*}
Furthermore, using the considerations from \S \ref{sec-chiral}, we get
\footnotesize
\begin{align*}
    &\Im\langle {\Id} \cdot \Psi \otimes \chi_*\Psi \rangle
    \\[0.1cm]
    =&\, -\frac{1}{2}  \left[ (d_+ - 2\ell_+)|\xi_+|^2 + (d_+ - 2\ell_+ + 2) |\eta_+|^2 + (d_- - 2\ell_-)|\xi_-|^2 + (d_- - 2\ell_- + 2)|\eta_-|^2 \right]\mathbf{e}^t\otimes\tau_3
    \\[0.1cm]
    & -\frac{1}{2} \left[ h_+(|\xi_+|^2 + |\eta_+|^2) + h_-(|\xi_-|^2 + |\eta_-|^2) \right] \mathbf{e}^t \otimes \tau_4
    \\[0.1cm]
    &+ \frac{1}{2} \left[ (d_+ - 2\ell_+)|\xi_+|^2 - (d_+ - 2\ell_+ + 2) |\eta_+|^2 - (d_- - 2\ell_-)|\xi_-|^2 + (d_- - 2\ell_- + 2)|\eta_-|^2 \right] \mathbf{e}^s \otimes\tau_3
    \\[0.1cm]
    &+ \frac{1}{2} \left[ h_+(|\xi_+|^2 - |\eta_+|^2) + h_- (- |\xi_-|^2 + |\eta_-|^2)\right] \mathbf{e}^s \otimes \tau_4
    \\[0.1cm]
    & + \left( \sqrt{\ell_+(d_+-\ell_++1)}\, \Re(\xi_+\conj{\eta}_+) - \sqrt{\ell_-(d_--\ell_-+1)}\, \Re(\xi_-\conj{\eta}_-) \right) (\mathbf{e}^\theta \otimes \tau_1 + \mathbf{e}^\varphi \otimes \tau_2)
    \\[0.1cm]
    & + \left( \sqrt{\ell_+(d_+-\ell_++1)}\, \Im(\xi_+\conj{\eta}_+) - \sqrt{\ell_-(d_--\ell_-+1)}\, \Im(\xi_-\conj{\eta}_-) \right)
    (\mathbf{e}^\theta \otimes \tau_2 - \mathbf{e}^\varphi \otimes \tau_1).
\end{align*}
\normalsize
It thus follows from \eqref{eq-ym-operator} that the coupled Yang-Mills operator is given by
\footnotesize
\begin{align*}
    \diff_\omega^\ast & F_\omega + \Re \langle \nabla_{\omega}\Phi \otimes \rho_*\Phi \rangle - \frac12 \Im \langle \id\cdot\Psi \otimes \chi_*\Psi \rangle
    \\[0.2cm]
    =& \mathbin{\phantom{+}} \bigg\{ - \frac{1}{r^2}e^{-\beta} (r^2 e^{-\alpha-\beta}a')' + \frac{2}{r^2} e^{-\alpha} a|z|^2 + \frac{qm}{2n} e^{-\alpha} \,\zeta^{(p,q)}_k |\phi|^2
    \\
    & \qquad + \frac{1}{4}  \left[ (d_+ - 2\ell_+)|\xi_+|^2 + (d_+ - 2\ell_+ + 2) |\eta_+|^2 + (d_- - 2\ell_-)|\xi_-|^2 + (d_- - 2\ell_- + 2)|\eta_-|^2 \right] \bigg\} \; \mathbf{e}^t \otimes \tau_3
    \\[0.2cm]
    &+ \bigg\{- \frac{1}{r^2}e^{-\beta} (r^2 e^{-\alpha-\beta}b')' -\frac{q}{2} e^{-\alpha} \,\zeta^{(p,q)}_k |\phi|^2
    +\frac{1}{4} \left[ h_+(|\xi_+|^2 + |\eta_+|^2) + h_-(|\xi_-|^2 + |\eta_-|^2) \right]
    \bigg\}\; \mathbf{e}^t \otimes \tau_4
    \\[0.2cm]
    & + \bigg\{\frac{2}{r^2} e^{-\beta}\Im(\conj{z}z') + \frac{qm}{2n} e^{-\beta} \,\Im(\conj{\phi}\phi')\\
    &\qquad - \frac{1}{4} \left[ (d_+ - 2\ell_+)|\xi_+|^2 - (d_+ - 2\ell_+ + 2) |\eta_+|^2 - (d_- - 2\ell_-)|\xi_-|^2 + (d_- - 2\ell_- + 2)|\eta_-|^2 \right] \bigg\} \; \mathbf{e}^s \otimes \tau_3
    \\[0.2cm]
    &- \bigg\{ \frac{q}{2} e^{-\beta} \,\Im(\conj{\phi}\phi') + \frac{1}{4} \left[ h_+(|\xi_+|^2 - |\eta_+|^2) + h_- (- |\xi_-|^2 + |\eta_-|^2)\right] \bigg\} \; \mathbf{e}^s \otimes \tau_4    
    \\[0.1cm]
    &- \frac{1}{r} \, \Re \bigg\{ e^{-\alpha-\beta}(e^{\alpha-\beta}z')' + e^{-2\alpha}a^2z + \frac{1}{r^2}z(n-|z|^2) 
    - \frac{\kappa z|\phi|^2}{4} 
    \\
    &\hspace{1.25cm} + \frac{r}{2}  \left[ \sqrt{\ell_+(d_+-\ell_++1)}\, \xi_+\conj{\eta}_+ - \sqrt{\ell_-(d_--\ell_-+1)}\, \xi_-\conj{\eta}_- \right] \bigg\} (\mathbf{e}^\theta \otimes \tau_1 + \mathbf{e}^\varphi \otimes \tau_2)
    \\[0.1cm]
    &- \frac{1}{r} \, \Im \bigg\{ e^{-\alpha-\beta}(e^{\alpha-\beta}z')' + e^{-2\alpha}a^2z + \frac{1}{r^2}z(n-|z|^2) 
    - \frac{\kappa z|\phi|^2}{4} 
    \\
    &\hspace{1.25cm} + \frac{r}{2}  \left[ \sqrt{\ell_+(d_+-\ell_++1)}\, \xi_+\conj{\eta}_+ - \sqrt{\ell_-(d_--\ell_-+1)}\, \xi_-\conj{\eta}_- \right] \bigg\} (\mathbf{e}^\theta \otimes \tau_2 - \mathbf{e}^\varphi \otimes \tau_1)
\end{align*}
\normalsize

\subsection{Einstein sector}

For the Einstein equations $-\Ein_g + \mathfrak{T} = 0$,
we note that the energy-momentum tensor is given by
\footnotesize
\begin{align*}
    \mathfrak{T}
    =&\;
    \mathfrak{T}^{\mathrm{YM}} + \mathfrak{T}^{\mathrm{H}} + \mathfrak{T}^{\mathrm{D}}
    \\[0.1cm]
    =&\,
     \bigg\{
    \frac12 e^{-2\alpha-2\beta} (a'{}^2 + b'{}^2)
    + \frac{1}{r^2}\left (e^{-2\alpha}a^2|z|^2 + e^{-2\beta}|z'|^2 
    + \frac{(n-|z|^2)^2+m^2}{2r^2} \right)
    \\
    &\quad 
    + \frac12\left( e^{-2\alpha}|\zeta_k^{(p,q)}\phi|^2 + e^{-2\beta}|\phi'|^2 + \frac{\kappa|z|^2|\phi|^2}{2r^2} + U(\phi)\right)
    \\
    &\quad
    -\frac12 e^{-\alpha}\left( \zeta^{(d_+,h_+)}_{\ell_+} |\xi_+|^2 +  \zeta^{(d_-,h_-)}_{\ell_-}|\xi_-|^2
    + \zeta^{(d_+,h_+)}_{\ell_+-1} |\eta_+|^2 +  \zeta^{(d_-,h_-)}_{\ell_--1}|\eta_-|^2 \right)
    \bigg\}\, \mathbf{e}^t \otimes \mathbf{e}^t
    \\[0.2cm]
    &+ \bigg\{ \frac{2}{r^2} e^{-\alpha-\beta} a\Im(\conj{z}z')  + \zeta_k^{(p,q)} e^{-\alpha-\beta}\Im(\conj{\phi}\phi') 
    \\
    &\qquad+ \frac12 e^{-\alpha}\bigg( \zeta^{(d_+,h_+)}_{\ell_+}|\xi_+|^2 - \zeta^{(d_-,h_-)}_{\ell_-}|\xi_-|^2 
    - \zeta^{(d_+,h_+)}_{\ell_+-1} |\eta_+|^2  + \zeta^{(d_-,h_-)}_{\ell_--1}|\eta_-|^2
    \bigg)  \bigg\} (\mathbf{e}^t \otimes \mathbf{e}^s + \mathbf{e}^s \otimes \mathbf{e}^t)
    \\[0.2cm]
    &+ \frac12 \bigg\{ - e^{-2\alpha-2\beta} (a'{}^2 + b'{}^2)
    + \frac{2}{r^2} \left( e^{-2\alpha}a^2|z|^2 + e^{-2\beta}|z'|^2 - \frac{(n-|z|^2)^2+m^2}{2r^2} \right)
    \\
    &\qquad + e^{-2\alpha}|\zeta_k^{(p,q)}\phi|^2 + e^{-2\beta}|\phi'|^2 - \frac{\kappa|z|^2|\phi|^2}{2r^2} - U(\phi)
    + e^{-\beta} \, \Im\left(\conj{\xi}_+ \xi_+' - \conj{\xi}_- \xi_-' - \conj{\eta}_+ \eta_+' + \conj{\eta}_- \eta_-' \right)
    \bigg\} \, \mathbf{e}^s \otimes \mathbf{e}^s
    \\[0.2cm]
    &+ \frac12 \bigg\{  e^{-2\alpha-2\beta} (a'{}^2 + b'{}^2)
    + \frac{(n-|z|^2)^2+m^2}{r^4}
    + e^{-2\alpha}|\zeta_k^{(p,q)}\phi|^2 - e^{-2\beta}|\phi'|^2 - U(\phi) 
    \\
    &\qquad - \frac{1}{r} \left[ \sqrt{\ell_+(d_+-\ell_++1)}\,\Re(z\conj{\xi}_+\eta_+) - \sqrt{\ell_-(d_--\ell_-+1)}\,\Re(z\conj{\xi}_-\eta_-) \right]
    \bigg\} (\mathbf{e}^\theta \otimes \mathbf{e}^\theta + \mathbf{e}^\varphi \otimes \mathbf{e}^\varphi),
\end{align*}
\normalsize
where we also use the considerations from \S \ref{sec-chiral}.
The equations are then easily deduced using the formula \eqref{eq-einstein-tensor-sph-sym} for the Einstein tensor.

\section{Schwarzschild coordinates}
\label{sec-schwarzschild-coordinates}

As already discussed in \S \ref{sec-coordinate-choices}, it is sometimes convenient to use the radius function $r$ as the independent variable, i.e.\ to set $s = r$. Then primes correspond to differentiation by $r$.
When studying black hole solutions in particular, is furthermore convenient to substitute the metric coefficients $\alpha$ and $\beta$ by $\delta$ and $M$ via
\begin{equation*}
    e^{2\alpha(r)} = e^{-2\delta(r)}\left(1 - \frac{2M(r)}{r}\right), \qquad 
    e^{2\beta(r)} = \left(1 - \frac{2M(r)}{r}\right)^{-1},
\end{equation*}
so that the metric can be written as
\begin{equation*}
    g = - e^{-2\delta(r)}\left(1 - \frac{2M(r)}{r}\right) \, \diff t^2 + \left(1 - \frac{2M(r)}{r}\right)^{-1} \, \diff r^2 + r^2 g_{\mathbb{S}^2}.
\end{equation*}
Then $M$ corresponds to the mass function, and e.g.\ the classical Schwarzschild black hole metric has $M \equiv \text{const}$ and $\delta \equiv 0$.
Denoting also
\begin{equation*}
    N(r) = 1-\frac{2M(r)}{r},
\end{equation*}
we note that
\begin{equation*}
    e^\alpha = e^{-\delta} N^{\frac12}, \qquad e^\beta = N^{-\frac12},
    \qquad
    \alpha' = -\delta' + \frac{M - rM'}{r(r-2M)}, \qquad \beta' = -\frac{M - rM'}{r(r-2M)},
\end{equation*}
and we can write the Einstein equations as
\begin{align*}
    M' =&\, \frac14 r^2 e^{2\delta} (a'{}^2 + b'{}^2)
        + \frac12 N^{-1} e^{2\delta} a^2 |z|^2 + \frac12 N |z'|^2 
        + \frac{(n-|z|^2)^2 + m^2}{4r^2}
        \\
        &+ \frac14 N^{-1} r^2 e^{2\delta} |\zeta^{(p,q)}_k \phi|^2
        + \frac14 r^2 N |\phi'|^2
        + \frac{\kappa}{8} |z|^2 |\phi|^2 
        + \frac14 r^2 U(\phi)
        \\
        &+ \frac14 N^{-\frac12} r^2 e^\delta 
        \left( \zeta^{(d_+,h_+)}_{\ell_+} |\xi_+|^2 +  \zeta^{(d_-,h_-)}_{\ell_-}|\xi_-|^2
        + \zeta^{(d_+,h_+)}_{\ell_+ - 1} |\eta_+|^2 +  \zeta^{(d_-,h_-)}_{\ell_- - 1}|\eta_-|^2 \right)
    \\[0.2cm]
    \delta' =& -\frac{1}{r} N^{-2} e^{2\delta} a^2 |z|^2 
        - \frac{1}{r} |z'|^2 
        - \frac12 r N^{-2} e^{2\delta} |\zeta^{(p,q)}_k \phi|^2 
        - \frac12 r|\phi'|^2
        \\
        & + \frac12 re^{\delta} N^{-\frac32} \left( \zeta^{(d_+,h_+)}_{\ell_+} |\xi_+|^2 +  \zeta^{(d_-,h_-)}_{\ell_-}|\xi_-|^2
        + \zeta^{(d_+,h_+)}_{\ell_+ - 1} |\eta_+|^2 +  \zeta^{(d_-,h_-)}_{\ell_- - 1}|\eta_-|^2 \right)
        \\
        &- \frac{1}{2} \sqrt{\ell_+ (d_+ - \ell_+ + 1)} \,N^{-1} \Re(z\conj{\xi}_+\eta_+)
        + \frac{1}{2} \sqrt{\ell_- (d_- - \ell_- + 1)} \,N^{-1} \Re(z\conj{\xi}_-\eta_-)
        \\
        &+ \frac12 m_\Y rN^{-1} \left( \Y_{k,\ell_+}^{\ell_-} \Re(\phi \conj{\xi}_- \xi_+) + \Y_{k,\ell_+-1}^{\ell_--1} \Re(\phi\conj{\eta}_-\eta_+) \right)
\end{align*}
Here, the equation for $\delta$ actually follows from the constraint (\ref{eq-einstein-constraint-2}) rather than from the second-order equation (\ref{eq-einstein-evo-ddt}), but one can easily verify that if $(M,\delta)$ solve the equations above then all three of the Einstein equations hold, cf.\ Remark \ref{rem-einstein-constraints}.
This has the clear benefit or reducing the number of derivatives (i.e.\ both of the equations above are of first order), but it comes at the cost of making them singular at marginally trapped spheres, i.e.\ those at which $N(r)=0$.
This is not a true singularity but rather a consequence of a "bad" choice of coordinates, i.e.\ at points like this the radius function $r$ is no longer monotone but has a stationary point.
This is the main reason why the other coordinate choices from \S \ref{sec-coordinate-choices} might be considered preferrable for analytic purposes.

\section{Propagation of the Einstein constraint}
\label{appendix-einstein-constraint}
In this appendix, we provide the necessary formulas to obtain \eqref{eq-einstein-constraint-propagation-equation}, which was used in the proof of Proposition \ref{prop-constraint-propagation}.
These computations are rather arduous so we do not present all the steps, but only the intermediate computations. 
That being said, all of these computations have been done by hand by the author.
Let
\begin{align*}
    A &= 1 - r^2 e^{-2\beta}\left(2\alpha'+\frac{r'}{r}\right) \frac{r'}{r}
    \\
    B &= -\frac12 r^2 e^{2\alpha-2\beta}(a'{}^2 + b'{}^2)
    \\
    C &= e^{-2\alpha} a^2 |z|^2 + e^{-2\beta}|z'|^2 - \frac{(n-|z|^2)^2 + m^2}{2r^2}
    \\
    D &= \frac12 r^2 e^{-2\alpha} |\zeta_k^{(p,q)}\phi|^2 + \frac12 r^2 e^{-2\beta}|\phi'|^2 - \frac{\kappa}{4}|z|^2 |\phi|^2 - \frac12 r^2 U(\phi)
    \\
    X &= -\frac12 r^2e^{-\alpha} \left[ \zeta_{\ell_+}^{(d_+, h_+)}|\xi_+|^2 + \zeta_{\ell_-}^{(d_-, h_-)}|\xi_-|^2 +  \zeta_{\ell_+ - 1}^{(d_+, h_+)}|\eta_+|^2 + \zeta_{\ell_- - 1}^{(d_-, h_-)}|\eta_-|^2 \right]
    \\
    Y &= r\left[ \sqrt{\ell_+(d_+ - \ell_+ - 1)} \Re(z\conj{\xi}_+\eta_+) - \sqrt{\ell_-(d_- - \ell_- - 1)} \Re(z\conj{\xi}_-\eta_-)  \right]
    \\
    Z &= -m_\Y r^2  \left[ \Y_{k,\ell_+}^{\ell_-} \Re(\phi\conj{\xi}_-\xi_+) + \Y_{k, \ell_+-1}^{\ell_--1} \Re(\phi\conj{\eta}_-\eta_+) \right]
\end{align*}
Then by the Einstein equations
\begin{align*}
    A' =& \, -2\left(\alpha' + \frac{r'}{r}\right)A
    \\
    &\, + \alpha' r^2 e^{2\alpha-2\beta} (a'{}^2 + b'{}^2) 
    - \frac{2r'}{r} (e^{-2\alpha} a^2 |z|^2 + e^{-2\beta} |z'|^2)
    + \alpha' \frac{(n-|z|^2)^2 + m^2}{r^2}
    \\
    &\, 
    + \frac{\kappa}{2} \left(\alpha' + \frac{r'}{r}\right) |z|^2 |\phi|^2
    - \frac{2r'}{r} r^2e^{-2\alpha} |\zeta_k^{(p,q)}\phi|^2
    + r^2 \left( \alpha' + \frac{2r'}{r} \right) U(\phi)
    \\
    &\, + \frac{r'}{r} r^2e^{-\alpha} \left[ \zeta_{\ell_+}^{(d_+, h_+)}|\xi_+|^2 + \zeta_{\ell_-}^{(d_-, h_-)}|\xi_-|^2 +  \zeta_{\ell_+ - 1}^{(d_+, h_+)}|\eta_+|^2 + \zeta_{\ell_- - 1}^{(d_-, h_-)}|\eta_-|^2 \right]
    \\
    &\, - \left( \alpha' + \frac{r'}{r} \right) \left[ r\sqrt{\ell_+(d_+ - \ell_+ - 1)} \Re(z\conj{\xi}_+\eta_+) - r\sqrt{\ell_-(d_- - \ell_- - 1)} \Re(z\conj{\xi}_-\eta_-)  \right]
    \\
    &\, + m_\Y r^2 \left( \alpha' + \frac{2r'}{r} \right) \left[ \Y_{k,\ell_+}^{\ell_-} \Re(\phi\conj{\xi}_-\xi_+) + \Y_{k, \ell_+-1}^{\ell_--1} \Re(\phi\conj{\eta}_-\eta_+) \right]
\end{align*}
\begin{align*}
    B' =& \, \frac{r'}{r} r^2 e^{-2\alpha-2\beta} (a'{}^2 + b{}^2) 
    - 2e^{-2\alpha} aa'|z|^2 
    -(\zeta_k^{(p,q)})' \zeta_k^{(p,q)} r^2e^{-2\alpha} |\phi|^2 
    \\
    &\,+ \frac{1}{2}r^2e^{-\alpha} \left[ (\zeta_{\ell_+}^{(d_+, h_+)})' |\xi_+|^2 + (\zeta_{\ell_-}^{(d_-, h_-)})' |\xi_-|^2 + (\zeta_{\ell_+ - 1}^{(d_+, h_+)})' |\eta_+|^2 + (\zeta_{\ell_- - 1}^{(d_-, h_-)})' |\eta_-|^2 \right] 
\end{align*}
\begin{align*}
    C' =& \, - 2\alpha' (e^{-2\alpha} a^2 |z|^2 + e^{-2\beta}|z'|^2) + \frac{r'}{r} \frac{(n-|z|^2)^2 + m^2}{r^2}
    \\
    &\, + 2e^{-2\alpha} aa' |z|^2 + \frac{\kappa}{2} |\phi|^2 \Re(\conj{z}z')
    \\
    &\, - r\sqrt{\ell_+(d_+ - \ell_+ - 1)} \Re(z'\conj{\xi}_+ \eta_+)
    + r\sqrt{\ell_-(d_- - \ell_- - 1)} \Re(z'\conj{\xi}_- \eta_-)
\end{align*}
\begin{align*}
    D' =& \,\left(\frac{r'}{r}-\alpha'\right) r^2e^{-2\alpha} |\zeta_k^{(p,q)}\phi|^2
        + r^2e^{-2\alpha}|\phi|^2 (\zeta_k^{(p,q)})' \zeta_k^{(p,q)}
    - \left( \alpha' + \frac{r'}{r} \right) r^2e^{-2\beta}|\phi'|^2
    \\
    &\,- \frac{\kappa}{2} |\phi|^2 \Re(\conj{z}z') - rr' U(\phi)
    + m_\Y r^2 \left[ \Y_{k,\ell_+}^{\ell_-} \Re(\phi' \xi_+ \conj{\xi}_-) + \Y_{k,\ell_+ - 1}^{\ell_- - 1}\Re(\phi' \eta_+ \conj{\eta}_-) \right] 
\end{align*}
\begin{align*}
    X' =&\, \alpha' r^2e^{-\alpha} \left[\zeta_{\ell_+}^{(d_+, h_+)}|\xi_+|^2 + \zeta_{\ell_-}^{(d_-, h_-)}|\xi_-|^2 +  \zeta_{\ell_+ - 1}^{(d_+, h_+)}|\eta_+|^2 + \zeta_{\ell_- - 1}^{(d_-, h_-)}|\eta_-|^2\right]
    \\
    &\, - \frac12 r^2 e^{-\alpha} \left[ (\zeta_{\ell_+}^{(d_+, h_+)})'|\xi_+|^2 + (\zeta_{\ell_-}^{(d_-, h_-)})'|\xi_-|^2 +  (\zeta_{\ell_+ - 1}^{(d_+, h_+)})'|\eta_+|^2 + (\zeta_{\ell_- - 1}^{(d_-, h_-)})' |\eta_-|^2 \right]
    \\
    &\, + re^{\beta-\alpha} \sqrt{\ell_+(d_+ - \ell_+ + 1)} \Im(z\conj{\xi}_+\eta_+) \left( -\zeta_{\ell_+}^{(d_+, h_+)} + \zeta_{\ell_+ - 1}^{(d_+, h_+)} \right)
    \\
    &\, + re^{\beta-\alpha} \sqrt{\ell_-(d_- - \ell_- + 1)} \Im(z\conj{\xi}_-\eta_-) \left( -\zeta_{\ell_-}^{(d_-, h_-)} + \zeta_{\ell_- - 1}^{(d_-, h_-)} \right)
    \\
    &\, - m_\Y \Y_{k,\ell_+}^{\ell_-} r^2 e^{\beta-\alpha} \left(\zeta_{\ell_+}^{(d_+,h_+)} + \zeta_{\ell_-}^{(d_-, h_-)}\right)  \Im(\phi \conj{\xi}_-\xi_+) 
    \\
    &\, + m_\Y \Y_{k,\ell_+ - 1}^{\ell_- - 1} r^2 e^{\beta-\alpha} \left(\zeta_{\ell_+ - 1}^{(d_+,h_+)} + \zeta_{\ell_- - 1}^{(d_-, h_-)}\right)  \Im(\phi \conj{\eta}_-\eta_+) 
\end{align*}
\begin{align*}
    Y' =& \, -\left( \alpha' + \frac{r'}{r} \right) Y
    + r\sqrt{\ell_+(d_+ - \ell_+ - 1)} \Re(z'\conj{\xi}_+\eta_+) - r\sqrt{\ell_-(d_- - \ell_- - 1)} \Re(z'\conj{\xi}_-\eta_-) 
    \\
    &\, + re^{\beta-\alpha}\sqrt{\ell_+(d_+-\ell_++1)} \, \Re(iz\conj{\xi}_+\eta_+)\left( \zeta_{\ell_+}^{(d_+, h_+)} + \zeta_{\ell_+ - 1}^{(d_+, h_+)} \right) 
    \\
    &\, + re^{\beta-\alpha}\sqrt{\ell_-(d_--\ell_-+1)} \, \Re(iz\conj{\xi}_-\eta_-)\left( \zeta_{\ell_-}^{(d_-, h_-)} + \zeta_{\ell_- - 1}^{(d_-, h_-)} \right) 
    \\
    &\, + m_\Y re^{\beta} \sqrt{\ell_+(d_+-\ell_++1)} \left[ \Y_{k,\ell_+}^{\ell_-} \Re(iz\phi\eta_+\conj{\xi}_-) + \Y_{k,\ell_+ - 1}^{\ell_- - 1} \Re(iz\conj{\phi}\conj{\xi}_+\eta_-) \right]
    \\
    &\, + m_\Y re^{\beta} \sqrt{\ell_-(d_--\ell_-+1)} \left[ \Y_{k,\ell_+}^{\ell_-} \Re(iz\conj{\phi}\conj{\xi}_+\eta_-) + \Y_{k,\ell_+ - 1}^{\ell_- - 1} \Re(iz\phi\conj{\xi}_-\eta_+) \right]
\end{align*}
\begin{align*}
    Z' =&\, -\alpha' Z - m_\Y r^2 \left[ \Y_{k,\ell_+}^{\ell_-} \Re(\phi' \conj{\xi}_-\xi_+) + \Y_{k, \ell_+-1}^{\ell_--1} \Re(\phi'\conj{\eta}_-\eta_+) \right]
    \\
    &\, - m_\Y \Y_{k,\ell_+}^{\ell_-} r^2 e^{\beta-\alpha} \Re(i\conj{\phi}\xi_-\conj{\xi}_+) \left( \zeta_{\ell_+}^{(d_+, h_+)} + \zeta_{\ell_-}^{(d_-, h_-)} \right)
    \\
    &\, - m_\Y \Y_{k,\ell_+ - 1}^{\ell_- - 1} r^2 e^{\beta-\alpha} \Re(i\phi\eta_+\conj{\eta}_-) \left( \zeta_{\ell_+ - 1}^{(d_+, h_+)} + \zeta_{\ell_- - 1}^{(d_-, h_-)} \right)
    \\
    &\, -m_\Y \sqrt{\ell_+(d_+-\ell_++1)} re^\beta \left[ \Y_{k,\ell_+}^{\ell_-} \Re(i\phi z\eta_+\conj{\xi}_-) - \Y_{k,\ell_+ - 1}^{\ell_- - 1} \Re(i\phi\conj{z}\xi_+\conj{\eta}_-) \right]
    \\
    &\, -m_\Y \sqrt{\ell_-(d_--\ell_-+1)} re^\beta \left[ \Y_{k,\ell_+}^{\ell_-} \Re(i\conj{\phi}z\eta_-\conj{\xi}_+) - \Y_{k,\ell_+ - 1}^{\ell_- - 1} \Re(i\conj{\phi}\conj{z}\xi_-\conj{\eta}_+) \right]
\end{align*}
Summing all these equations, we get that
\begin{equation*}
    (A+B+C+D+X+Y+Z)' = -2\left( \alpha' + \frac{r'}{r} \right)(A+B+C+D+X+Y+Z),
\end{equation*}
which is precisely \eqref{eq-einstein-constraint-propagation-equation}.

\end{appendix}

\bibliography{refs}
\bibliographystyle{amsplain}

\end{document}